\documentclass[11pt,envcountsect]{llncs}
\usepackage[margin=1in]{geometry}
\usepackage{amsmath,amssymb}
\usepackage[final]{graphicx}
\usepackage{algorithm}
\usepackage[noend]{algpseudocode}
\usepackage{microtype}
\floatname{algorithm}{Algorithm}

\def\cut{\cap}
\def\eps{\varepsilon}

\spnewtheorem*{proofsketch}{Sketch of proof}{\itshape}{\rmfamily}

\def\kCSP{{\sc Max}-$k$-CSP}
\def\kSAT{{\sc Max}-$k$-SAT}
\def\MC{{\sc Max}-CUT}
\def\MDC{{\sc Max}-DICUT}
\def\kDense{$k$-{\sc Densest Subgraph}}
\def\deg{\mathrm{deg}}

\def\eps{\varepsilon}
\def\e{\epsilon}
\def\rb{\bar{\rho}}
\def\tb{\bar{\tau}}
\def\Prob{\mathbb{P}}
\def\Exp{\mathbb{E}}
\title{Sub-exponential Approximation Schemes for CSPs:\\ from Dense to Almost
Sparse%
\thanks{This research was supported by the project AlgoNow, co-financed by the European Union (European Social Fund - ESF) and Greek national funds, through the Operational Program ``Education and Lifelong Learning'' of the National Strategic Reference Framework (NSRF) - Research Funding Program: THALES, investing in knowledge society through the European Social Fund.}}

\author{Dimitris Fotakis\inst{1} \and Michael Lampis\inst{2} \and Vangelis Th. Paschos\inst{2}}

\institute{%
School of Electrical and Computer Engineering, National Technical University of Athens, Greece\\
\email{fotakis@cs.ntua.gr}
\and
LAMSADE, Universit\'{e} Paris Dauphine, France\\
\email{michail.lampis@dauphine.fr}, \email{paschos@lamsade.dauphine.fr}
}

\begin{document}

\maketitle

\begin{abstract}

It has long been known, since the classical work of (Arora, Karger, Karpinski,
JCSS~99), that \MC\ admits a PTAS on dense graphs, and more generally, \kCSP\
admits a PTAS on ``dense'' instances with $\Omega(n^k)$ constraints. In this
paper we extend and generalize their exhaustive sampling approach, presenting a
framework for $(1-\eps)$-approximating any \kCSP\ problem in
\emph{sub-exponential} time while significantly relaxing the denseness
requirement on the input instance.

Specifically, we prove that for any constants $\delta \in
(0, 1]$ and $\eps > 0$, we can approximate \kCSP\ problems with
$\Omega(n^{k-1+\delta})$ constraints within a factor of $(1-\eps)$ in time
$2^{O(n^{1-\delta}\ln n /\eps^3)}$. The framework is quite general and includes
classical optimization problems, such as \MC, {\sc Max}-DICUT, \kSAT, and (with
a slight extension) $k$-{\sc Densest Subgraph}, as special cases. For \MC\ in
particular (where $k=2$), it gives an approximation scheme that runs in time
sub-exponential in $n$ even for ``almost-sparse'' instances (graphs with
$n^{1+\delta}$ edges).

We prove that our results are essentially best possible, assuming the ETH.
First, the density requirement cannot be relaxed further: there exists a
constant $r < 1$ such that for all $\delta > 0$, \kSAT\ instances with
$O(n^{k-1})$ clauses cannot be approximated within a ratio better than $r$ in
time $2^{O(n^{1-\delta})}$.  Second, the running time of our algorithm is
almost tight \emph{for all densities}. Even for \MC\ there exists $r<1$ such
that for all $\delta'
> \delta >0$, \MC\ instances with $n^{1+\delta}$ edges cannot be approximated
 within a ratio better than $r$ in time $2^{n^{1-\delta'}}$.

\end{abstract}

%\thispagestyle{empty}%
%\setcounter{page}{0}%
%\newpage
%\pagestyle{plain}
%\pagenumbering{arabic}

\section{Introduction}
\label{s:intro}

The complexity of Constraint Satisfaction Problems (CSPs) has long played a
central role in theoretical computer science and it quickly became evident that
almost all interesting CSPs are NP-complete \cite{S78}.  Thus, since
approximation algorithms are one of the standard tools for dealing with NP-hard
problems, the question of approximating the corresponding optimization problems
({\sc Max}-CSP) has attracted significant interest over the years \cite{T10}.
Unfortunately, most CSPs typically resist this approach: not only are they
APX-hard \cite{KSW97}, but quite often the best polynomial-time approximation
ratio we can hope to achieve for them is that guaranteed by a trivial random
assignment \cite{H01}. This striking behavior is often called
\emph{approximation resistance}.

Approximation resistance and other APX-hardness results were originally
formulated in the context of \emph{polynomial-time} approximation. It would
therefore seem that one conceivable way for working around such barriers could
be to consider approximation algorithms running in super-polynomial time, and
indeed super-polynomial approximation for NP-hard problems is a topic that has
been gaining more attention in the literature recently
\cite{CLN13,BEP09,BCEP13,CKW09,CP10,CPW11}.  Unfortunately, the existence of
quasi-linear PCPs with small soundness error, first given in the work of
Moshkovitz and Raz \cite{MR10}, established that approximation resistance is a
phenomenon that carries over even to \emph{sub-exponential} time approximation,
essentially ``killing'' this approach for CSPs.  For instance, we now know that
if, for any $\eps>0$, there exists an algorithm for {\sc Max}-3-SAT with ratio
$7/8+\eps$ running in time $2^{n^{1-\eps}}$ this would imply the existence of a
sub-exponential \emph{exact} algorithm for 3-SAT, disproving the Exponential
Time Hypothesis (ETH).  It therefore seems that sub-exponential time
does not improve the approximability of CSPs, or put another way, for many CSPs
obtaining a very good approximation ratio requires almost as much time as
solving the problem exactly.

Despite this grim overall picture, many positive approximation results for CSPs
have appeared over the years, by taking advantage of the special structure of
various classes of instances. One notable line of research in this vein is the
work on the approximability of \emph{dense} CSPs, initiated by Arora, Karger
and Karpinski \cite{AKK99} and independently by de la Vega \cite{V96}.  The
theme of this set of results is that the problem of maximizing the number of
satisfied constraints in a CSP instance with arity $k$ (\kCSP) becomes
significantly easier if the instance contains $\Omega(n^k)$ constraints. More
precisely, it was shown in \cite{AKK99} that \kCSP\ admits a
\emph{polynomial-time approximation scheme} (PTAS) on dense instances, that is,
an algorithm which for any constant $\eps>0$ can in time polynomial in $n$
produce an assignment that satisfies $(1-\eps)\mathrm{OPT}$ constraints.
Subsequent work produced a stream of positive
\cite{VK00,BVK03,AVKK03,CKSV12,CKSV11,FK96,AFK02,DFJ98,II05} (and some negative
\cite{VK99,AA07}) results on approximating CSPs which are in general APX-hard,
showing that dense instances form an island of tractability where many
optimization problems which are normally APX-hard admit a PTAS.

\noindent\textbf{Our contribution}: The main goal of this paper is to use the
additional power afforded by sub-exponential time to extend this island of
tractability as much as possible.  To demonstrate the main result, consider a
concrete CSP such as {\sc Max}-3-SAT.  As mentioned, we know that
sub-exponential time does not in general help us approximate this problem: the
best ratio achievable in, say, $2^{\sqrt{n}}$ time is still 7/8.  On the other
hand, this problem admits a PTAS on instances with $\Omega(n^3)$ clauses. This
density condition is, however, rather strict, so the question we would like to
answer is the following: Can we efficiently approximate a larger (and more
sparse) class of instances while using sub-exponential time?

In this paper we provide a positive answer to this question, not just for {\sc
Max}-3-SAT, but also for any \kCSP\ problem.  Specifically, we show that for
any constants $\delta\in (0,1]$, $\eps>0$ and integer $k\ge 2$, there is an
algorithm which achieves a $(1-\eps)$ approximation of \kCSP\ instances with
$\Omega(n^{k-1+\delta})$ constraints in time $2^{O(n^{1-\delta}\ln n
/\eps^3)}$. A notable special case of this result is for $k=2$, where the input
instance can be described as a graph. For this case, which contains classical
problems such as \MC, our algorithm gives an approximation scheme running in
time $2^{O(\frac{n}{\Delta}\ln n/\eps^3)}$ for graphs with average degree
$\Delta$. In other words, this is an approximation scheme that runs in time
\emph{sub-exponential in $n$} even for almost sparse instances where the
average degree is $\Delta = n^\delta$ for some small $\delta>0$. More
generally, our algorithm provides a trade-off between the time available and
the density of the instances we can handle. For graph problems ($k=2$) this
trade-off covers the whole spectrum from dense to almost sparse instances,
while for general \kCSP, it covers instances where the number of constraints
ranges from $\Theta(n^{k})$ to $\Theta(n^{k-1})$.

\noindent\textbf{Techniques}: The algorithms in this paper are an extension and
generalization of the \emph{exhaustive sampling} technique given by Arora,
Karger and Karpinski \cite{AKK99}, who introduced a framework of smooth
polynomial integer programs to give a PTAS for dense \kCSP. The basic idea of
that work can most simply be summarized for \MC. This problem can be recast  as
the problem of maximizing a quadratic function over $n$ boolean variables.
This is of course a hard problem, but suppose that we could somehow ``guess''
for each vertex how many of its neighbors belong in each side of the cut. This
would make the quadratic problem linear, and thus much easier. The main
intuition now is that, if the graph is dense, we can take a sample of $O(\log
n)$ vertices and guess their partition in the optimal solution. Because every
non-sample vertex will have ``many'' neighbors in this sample, we can with high
confidence say that we can estimate the fraction of neighbors on each side for
all vertices.  The work of de la Vega \cite{V96} uses exactly this algorithm
for \MC, greedily deciding the vertices outside the sample. The work of
\cite{AKK99} on the other hand pushed this idea to its logical conclusion,
showing that it can be applied to degree-$k$ polynomial optimization problems,
by recursively turning them into linear programs whose coefficients are
estimated from the sample. The linear programs are then relaxed to produce
fractional solutions, which can be rounded back into an integer solution to the
original problem.

On a very high level, the approach we follow in this paper retraces the steps
of \cite{AKK99}: we formulate \kCSP\ as a degree-$k$ polynomial maximization
problem; we then recursively decompose the degree-$k$ polynomial problem into
lower-degree polynomial optimization problems, estimating the coefficients by
using a sample of variables for which we try all assignments; the result of
this process is an integer linear program, for which we obtain a fractional
solution in polynomial time; we then perform randomized rounding to obtain an
integer solution that we can use for the original problem.

The first major difference between our approach and \cite{AKK99} is of course
that we need to use a larger sample. This becomes evident if one considers \MC\
on graphs with average degree $\Delta$. In order to get the sampling scheme to
work we must be able to guarantee that each vertex outside the sample has
``many'' neighbors inside the sample, so we can safely estimate how many of
them end up on each side of the cut. For this, we need a sample of size at
least $n\log n/\Delta$. Indeed, we use a sample of roughly this size, and exhausting
all assignments to the sample is what dominates the running time of our
algorithm. As we argue later, not only is the sample size we use essentially
tight, but more generally the running time of our algorithm is essentially
optimal (under the ETH).

Nevertheless, using a larger sample is not in itself sufficient to extend the
scheme of \cite{AKK99} to non-dense instances. As observed in \cite{AKK99} ``to
achieve a multiplicative approximation for dense instances it suffices to
achieve an additive approximation for the nonlinear integer programming
problem''. In other words, one of the basic ingredients of the analysis of
\cite{AKK99} is that additive approximation errors of the order $\eps n^k$ can
be swept under the rug, because we know that in a dense instance the optimal
solution has value $\Omega(n^k)$. This is \emph{not} true in our case, and we
are therefore forced to give a more refined analysis of the error of our
scheme, independently bounding the error introduced in the first step
(coefficient estimation) and the last (randomized rounding).

A further complication arises when considering \kCSP\ for $k>2$. The scheme of
\cite{AKK99} recursively decomposes such dense instances into lower-order
polynomials which retain the same ``good'' properties. This seems much harder
to extend to the non-dense case, because intuitively if we start from a
non-dense instance the decomposition could end up producing some dense and
some sparse sub-problems. Indeed we present a scheme that approximates \kCSP\
with $\Omega(n^{k-1+\delta})$ constraints, but does not seem to extend to
instances with fewer than $n^{k-1}$ constraints.  As we will see, there seems
to be a fundamental complexity-theoretic justification explaining exactly why
this decomposition method cannot be extended further.

To ease presentation, we first give all the details of our scheme for the
special case of \MC\ in Section \ref{s:maxcut}. We then present the full
framework for approximating \emph{smooth polynomials} in Section \ref{s:pip};
this implies the approximation result for \kSAT\ and more generally \kCSP.  We
then show in Section \ref{s:kdense} that it is possible to extend our framework
to handle \kDense, a problem which can be expressed as the maximization of a
polynomial subject to linear constraints.  For this problem we obtain an
approximation scheme which, given a graph with average degree $\Delta=n^\delta$
gives a $(1-\eps)$ approximation in time $2^{O(n^{1-\delta/3}\ln n/\eps^3)}$.
Observe that this extends the result of \cite{AKK99} for this problem not only
in terms of the density of the input instance, but also in terms of $k$ (the
result of \cite{AKK99} required that $k=\Omega(n)$).

\noindent\textbf{Hardness}: What makes the results of this paper more
interesting is that we can establish that in many ways they are essentially
best possible, if one assumes the ETH. In particular, there are at least two
ways in which one may try to improve on these results further: one would be to
improve the running time of our algorithm, while another would be to extend the
algorithm to the range of densities it cannot currently handle. In Section
\ref{s:lower} we show that both of these approaches would face significant
barriers. Our starting point is the fact that (under ETH) it takes exponential
time to approximate \MC\ arbitrarily well on sparse instances, which is a
consequence of the existence of quasi-linear PCPs. By manipulating such \MC\
instances, we are able to show that for \emph{any} average degree
$\Delta=n^{\delta}$ with $\delta<1$ the time needed to approximate \MC\
arbitrarily well almost matches the performance of our algorithm. Furthermore,
starting from sparse \MC\ instances, we can produce instances of \kSAT\ with
$O(n^{k-1})$ clauses while preserving hardness of approximation. This gives a
complexity-theoretic justification for our difficulties in decomposing \kCSP\
instances with less than $n^{k-1}$ constraints.
\section{Notation and Preliminaries}
\label{s:prelim}

An $n$-variate degree-$d$ polynomial $p(\vec{x})$ is \emph{$\beta$-smooth} \cite{AKK99}, for some constant $\beta \geq 1$, if for every $\ell \in \{ 0, \ldots, d\}$, the absolute value of each coefficient of each degree-$\ell$ monomial in the expansion of $p(\vec{x})$ is at most $\beta n^{d - \ell}$.
An $n$-variate degree-$d$ $\beta$-smooth polynomial $p(\vec{x})$ is \emph{$\delta$-bounded}, for some constant $\delta \in (0, 1]$, if for every $\ell$, the sum, over all degree-$\ell$ monomials in $p(\vec{x})$, of the absolute values of their coefficients is $O(\beta n^{d-1+\delta})$. Therefore, for any $n$-variate degree-$d$ $\beta$-smooth $\delta$-bounded polynomial $p(\vec{x})$ and any $\vec{x} \in \{ 0, 1\}^n$, $|p(\vec{x})| = O(d \beta n^{d-1+\delta})$.

Throughout this work, we treat $\beta$, $\delta$ and $d$ as fixed constants and express the running time of our algorithm as a function of $n$, i.e., the number of variables in $p(\vec{x})$. %Nevertheless, the size of the input is determined by the total number $m$ of monomials in the expansion of $p(\vec{x})$.

\noindent{\bf Optimization Problem.}
Our approximation schemes for almost sparse instances of \MC, \kSAT, and \kCSP\ are obtained by reducing them to the following problem: Given an $n$-variate $d$-degree $\beta$-smooth $\delta$-bounded polynomial $p(\vec{x})$, we seek a binary vector $\vec{x}^\ast \in \{0, 1\}^n$ that maximizes $p$, i.e., for all binary vectors $\vec{y} \in \{0, 1\}^n$, $p(\vec{x}^\ast) \geq p(\vec{y})$.

\noindent{\bf Polynomial Decomposition and General Approach.}
As in \cite[Lemma~3.1]{AKK99}, our general approach is motivated by the fact that any $n$-variate $d$-degree $\beta$-smooth polynomial $p(\vec{x})$ can be naturally decomposed into a collection of $n$ polynomials $p_j(\vec{x})$. Each of them has degree $d-1$ and at most $n$ variables and is $\beta$-smooth.
\begin{proposition}[\cite{AKK99}]\label{pr:decomposition}
Let $p(\vec{x})$ be any $n$-variate degree-$d$ $\beta$-smooth polynomial. Then, there exist a constant $c$ and degree-$(d-1)$ $\beta$-smooth polynomials $p_j(\vec{x})$ such that
\( p(\vec{x}) = c + \sum_{j = 1}^n x_j p_j(\vec{x}) \).
\end{proposition}
\begin{proof} The proposition is shown in \cite[Lemma~3.1]{AKK99}. We prove it
here just for completeness. Each polynomial $p_j(\vec{x})$ is obtained from
$p(\vec{x})$ if we keep only the monomials with variable $x_j$ and pull $x_j$
out, as a common factor. The constant $c$ takes care of the constant term in
$p(\vec{x})$. Each monomial of degree $\ell$ in $p(\vec{x})$ becomes a monomial
of degree $\ell-1$ in $p_j(\vec{x})$, which implies that the degree of
$p_j(\vec{x})$ is $d-1$. Moreover, by the $\beta$-smoothness condition, the
coefficient $t$ of each degree-$\ell$ monomial in $p(\vec{x})$ has $|t| \leq
\beta n^{d - \ell}$. The corresponding monomial in $p_j(\vec{x})$ has degree
$\ell-1$ and the same coefficient $t$ with $|t| \leq \beta n^{d - 1 -
(\ell-1)}$. Therefore, if $p(\vec{x})$ is $\beta$-smooth, each $p_j(\vec{x})$
is also $\beta$-smooth.  \qed\end{proof}
\noindent{\bf Graph Optimization Problems.}
Let $G(V, E)$ be a (simple) graph with $n$ vertices and $m$ edges. For each vertex $i \in V$, $N(i)$ denotes $i$'s neighborhood in $G$, i.e., $N(i) = \{ j \in V: \{i, j\} \in E\}$. We let $\deg(i) = |N(i)|$ be the degree of $i$ in $G$ and $\Delta = 2|E|/n$ denote the average degree of $G$.
We say that a graph $G$ is \emph{$\delta$-almost sparse}, for some constant $\delta \in (0, 1]$, if $m = \Omega(n^{1+\delta})$ (and thus, $\Delta = \Omega(n^\delta)$).

In \MC, we seek a partitioning of the vertices of $G$ into two sets $S_0$ and
$S_1$ so that the number of edges with endpoints in $S_0$ and $S_1$ is
maximized. If $G$ has $m$ edges, the number of edges in the optimal cut is at
least $m/2$.
%
%\MDC\ is the directed version of \MC, where $G(V, E)$ is a directed graph and we seek to maximize the number of edges $(u, v) \in E$ with $u \in S_0$ and $v \in S_1$. Similarly to \MC, if $|E| = m$, the number of edges in the optimal directed cut is at least $m/4$.

In \kDense, given an undirected graph $G(V, E)$, we seek a subset $C$ of $k$ vertices so that the induced subgraph $G[C]$ has a maximum number of edges. 

\noindent{\bf Constraint Satisfaction Problems.}
An instance of (boolean) \kCSP\ with $n$ variables consists of $m$ boolean constraints $f_1, \ldots, f_m$, where each $f_j : \{ 0, 1\}^k \to \{0, 1\}$ depends on $k$ variables and is satisfiable, i.e., $f_j$ evaluates to $1$ for some truth assignment. We seek a truth assignment to the variables that maximizes the number of satisfied constraints. \kSAT\ is a special case of \kCSP\ where each constraint $f_j$ is a disjunction of $k$ literals. An averaging argument implies that the optimal assignment of a \kCSP\ (resp. \kSAT) instance with $m$ constraints satisfies at least $2^{-k} m$ (resp. $(1-2^{-k})m$) of them. We say that an instance of \kCSP\ is \emph{$\delta$-almost sparse}, for some constant $\delta \in (0, 1]$, if the number of constraints is $m = \Omega(n^{k-1+\delta})$.

Using standard arithmetization techniques (see e.g., \cite[Sec.~4.3]{AKK99}), we can reduce any instance of \kCSP\ with $n$ variables to an $n$-variate degree-$k$ polynomial $p(\vec{x})$ so that the optimal truth assignment for \kCSP\ corresponds to a maximizer $\vec{x}^\ast \in \{0, 1\}$ of $p(\vec{x})$ and the value of the optimal \kCSP\ solution is equal to $p(\vec{x}^\ast)$. Since each $k$-tuple of variables can appear in at most $2^k$ different constraints, $p(\vec{x})$ is $\beta$-smooth, for $\beta \in [1, 4^k]$, and has at least $m$ and at most $4^k m$ monomials. Moreover, if the instance of \kCSP\ has $m = \Theta(n^{k-1+\delta})$ constraints, then $p(\vec{x})$ is $\delta$-bounded and its maximizer $\vec{x}^\ast$ has $p(\vec{x}^\ast) = \Omega(n^{k-1+\delta})$.

\noindent{\bf Notation and Terminology.}
An algorithm has \emph{approximation ratio} $\rho \in (0, 1]$ (or is \emph{$\rho$-approximate}) if for all instances, the value of its solution is at least $\rho$ times the value of the optimal solution.

For graphs with $n$ vertices or CSPs with $n$ variables, we say that an event $E$ happens with high probability (or whp.), if $E$ happens with probability at least $1-1/n^c$, for some constant $c \geq 1$.

For brevity and clarity, we sometimes write $\alpha \in (1\pm \e_1) \beta \pm \e_2 \gamma$, for some constants $\e_1, \e_2 > 0$, to denote that $(1-\e_1)\beta - \e_2 \gamma \leq \alpha \leq (1+\e_1)\beta + \e_2 \gamma$.

\section{Approximating \MC\ in Almost Sparse Graphs}
\label{s:maxcut}

In this section, we apply our approach to \MC, which serves as a convenient example and allows us to present the intuition and the main ideas.

The \MC\ problem in a graph $G(V, E)$ is equivalent to maximizing, over all binary vectors $\vec{x} \in \{0, 1\}^n$, the following $n$-variate degree-$2$ $2$-smooth polynomial
\[ p(\vec{x}) = \sum_{\{i, j\} \in E} (x_i (1 - x_j) + x_j (1 - x_i)) \]
Setting a variable $x_i$ to $0$ indicates that the corresponding vertex $i$ is assigned to the left side of the cut, i.e., to $S_0$, and setting $x_i$ to $1$ indicates that vertex $i$ is assigned to the right side of the cut, i.e., to $S_1$.
We assume that $G$ is $\delta$-almost sparse and thus, has $m = \Omega(n^{1+\delta})$ edges and average degree $\Delta = \Omega(n^\delta)$.
Moreover, if $m = \Theta(n^{1+\delta})$, $p(\vec{x})$ is $\delta$-bounded, since for each edge $\{i, j\} \in E$, the monomial $x_ix_j$ appears with coefficient $-2$ in the expansion of $p$, and for each vertex $i \in V$, the monomial $x_i$ appears with coefficient $\deg(i)$ in the expansion of $p$. Therefore, for $\ell \in \{1, 2\}$, the sum of the absolute values of the coefficients of all monomials of degree $\ell$ is at most $2m = O(n^{1+\delta})$. 

Next, we extend and generalize the approach of \cite{AKK99} and show how to $(1-\eps)$-approximate the optimal cut, for any constant $\eps > 0$, in time $2^{O(n\ln n/(\Delta \eps^3))}$ (see Theorem~\ref{th:maxcut}). The running time is subexponential in $n$, if $G$ is $\delta$-almost sparse. 

\subsection{Outline and Main Ideas}
\label{s:cut_main}

Applying Proposition~\ref{pr:decomposition}, we can write the smooth polynomial $p(\vec{x})$ as
\begin{equation}\label{eq:cut_decomp}
p(\vec{x}) = \sum_{j \in V} x_j (\deg(j) - p_j(\vec{x}))\,,
\end{equation}
where $p_j(\vec{x}) = \sum_{i \in N(j)} x_i$ is a degree-$1$ $1$-smooth polynomial that indicates how many neighbors of vertex $j$ are in $S_1$ in the solution corresponding to $\vec{x}$. The key observation, due to \cite{AKK99}, is that if we have a good estimation $\rho_j$ of the value of each $p_j$ at the optimal solution $\vec{x}^\ast$, then approximate maximization of $p(\vec{x})$ can be reduced to the solution of the following Integer Linear Program:
\begin{alignat*}{3}
& &\max \sum_{j \in V} &y_j (\deg(j) - \rho_j) & & \tag{IP}\\
&\mathrm{s.t.}\quad &
(1-\e_1) \rho_j - \e_2 \Delta \leq \sum_{i \in N(j)} &y_i \leq (1+\e_1) \rho_j + \e_2 \Delta \quad & \forall &j \in V\\
& & &y_j \in \{0, 1\} &\forall & j \in V
\end{alignat*}
The constants $\e_1, \e_2 > 0$ and the estimations $\rho_j \geq 0$ are computed so that the optimal solution $\vec{x}^\ast$ is a feasible solution to (IP). We always assume wlog. that $0 \leq \sum_{i \in N(j)} y_i \leq \deg(j)$, i.e., we let the lhs of the $j$-th constraint be $\max\{ (1-\e_1) \rho_j - \e_2 \Delta, 0 \}$ and the rhs be $\min\{ (1+\e_1) \rho_j + \e_2 \Delta, \deg(j) \}$. Clearly, if $\vec{x}^\ast$ is a feasible solution to (IP), it remains a feasible solution after this modification. We let (LP) denote the Linear Programming relaxation of (IP), where each $y_j \in [0, 1]$.

The first important observation is that for any $\e_1, \e_2 > 0$, we can compute estimations $\rho_j$, by exhaustive sampling, so that $\vec{x}^\ast$ is a feasible solution to (IP) with high probability (see Lemma~\ref{l:cut_sampling}). The second important observation is that the objective value of any feasible solution $\vec{y}$ to (LP) is close to $p(\vec{y})$ (see Lemma~\ref{l:cut_approx}). Namely, for any feasible solution $\vec{y}$, $\sum_{j \in V} y_j (\deg(j) - \rho_j) \approx p(\vec{y})$.

Based on these observations, the approximation algorithm performs the following steps:
\begin{enumerate}
\item We guess a sequence of estimations $\rho_1, \ldots, \rho_n$, by exhaustive sampling, so that $\vec{x}^\ast$ is a feasible solution to the resulting (IP) (see Section~\ref{s:cut_sampling} for the details).
\item We formulate (IP) and find an optimal fractional solution $\vec{y}^\ast$ to (LP).
\item We obtain an integral solution $\vec{z}$ by applying randomized rounding to $\vec{y}^\ast$ (and the method of conditional probabilities, as in \cite{RT87,Rag88}).
\end{enumerate}
To see that this procedure indeed provides a good approximation to $p(\vec{x}^\ast)$, we observe that:
\begin{equation}\label{eq:cut_est}
%\[
  p(\vec{z}) \approx \sum_{j \in V} z_j (\deg(j) - \rho_j) \approx
                      \sum_{j \in V} y^\ast_j (\deg(j) - \rho_j) \geq
                      \sum_{j \in V} x^\ast_j (\deg(j) - \rho_j) \approx
                      p(\vec{x}^\ast)\,,
%\]
\end{equation}
The first approximation holds because $\vec{z}$ is an (almost) feasible solution to (IP) (see Lemma~\ref{l:cut_approx2}), the second approximation holds because the objective value of $\vec{z}$ is a good approximation to the objective value of $\vec{y}^\ast$, due to randomized rounding, the inequality holds because $\vec{x}^\ast$ is a feasible solution to (LP) and the final approximation holds because $\vec{x}^\ast$ is a feasible solution to (IP).

In Sections~\ref{s:cut_linearization}~and~\ref{s:cut_rounding}, we make the notion of approximation precise so that $p(\vec{z}) \geq (1-\eps) p(\vec{x}^\ast)$. As for the running time, it is dominated by the time required for the exhaustive-sampling step. Since we do not know $\vec{x}^\ast$, we need to run the steps (2) and (3) above for every sequence of estimations produced by exhaustive sampling. So, the outcome of the approximation scheme is the best of the integral solutions $\vec{z}$ produced in step (3) over all executions of the algorithm. In Section~\ref{s:cut_sampling}, we show that a sample of size $O(n \ln n/\Delta)$ suffices for the computation of estimations $\rho_j$ so that $\vec{x}^\ast$ is a feasible solution to (IP) with high probability. If $G$ is $\delta$-almost sparse, the sample size is sublinear in $n$ and the running time is subexponential in $n$.

\subsection{Obtaining Estimations $\rho_j$ by Exhaustive Sampling}
\label{s:cut_sampling}

To obtain good estimations $\rho_j$ of the values $p_j(\vec{x}^\ast) = \sum_{i \in N(j)} x_i^\ast$, i.e., of the number of $j$'s neighbors in $S_1$ in the optimal cut, we take a random sample $R \subseteq V$ of size $\Theta(n \ln n / \Delta)$ and try exhaustively all possible assignments of the vertices in $R$ to $S_0$ and $S_1$. If $\Delta = \Omega(n^\delta)$, we have $2^{O(n\ln n / \Delta)} = 2^{O(n^{1-\delta} \ln n)}$ different assignments. For each assignment, described by a $0/1$ vector $\vec{x}$ restricted to $R$, we compute an estimation $\rho_j = (n / |R|) \sum_{i \in N(j) \cut R} x_i$, for each vertex $j \in V$, and run the steps (2) and (3) of the algorithm above. Since we try all possible assignments, one of them agrees with $\vec{x}^\ast$ on all vertices of $R$. So, for this assignment, the estimations computed are $\rho_j = (n / |R|) \sum_{i \in N(j) \cut R} x^\ast_i$. 
The following shows that for these estimations, we have that $p_j(\vec{x}^\ast) \approx \rho_j$ with high probability.
%The following is an immediate corollary of Lemma~\ref{l:sampling} (for $\beta = 1$, $d = 2$ and $q = 0$, and with $\Delta$ instead of $n^\delta$), proven in Section~\ref{s:sampling}. If we focus on the estimations $\rho_j = (n / |R|) \sum_{i \in N(j) \cut R} x^\ast_i$, we have that $p_j(\vec{x}^\ast) \approx \rho_j$ with high probability.
\begin{lemma}\label{l:cut_sampling}
Let $\vec{x}$ be any binary vector. For all $\alpha_1, \alpha_2 > 0$, we let $\gamma = \Theta(1/(\alpha^2_1 \alpha_2))$ and let $R$ be a multiset of $r = \gamma n \ln n / \Delta$ vertices chosen uniformly at random with replacement from $V$. For any vertex $j$, if $\rho_j = (n / r) \sum_{i \in N(j) \cut R} x_i$ and $\hat{\rho}_j = \sum_{i \in N(j)} x_i$, with probability at least $1 - 2/n^{3}$,
\begin{equation}\label{eq:cut_sample_cor}
 (1-\alpha_1)\hat{\rho}_j - (1-\alpha_1)\alpha_2 \Delta \leq \rho_j \leq
 (1+\alpha_1)\hat{\rho}_j + (1+\alpha_1)\alpha_2 \Delta
\end{equation}
\end{lemma}
\begin{proofsketch} If $\hat{\rho}_j = \Omega(\Delta)$, the neighbors of $j$
are well-represented in the random sample $R$ whp., because $|R| = \Theta(n\ln
n/\Delta)$. Therefore, $|\hat{\rho}_j - \rho_j| \leq \alpha_1\hat{\rho}_j$
whp., by Chernoff bounds. If $\hat{\rho}_j = o(\Delta)$, the lower bound in
(\ref{eq:cut_sample_cor}) becomes trivial, since it is non-positive, while
$\rho_j \geq 0$. As for the upper bound, we increase some $x_i$ to $x'_i \in
[0, 1]$, so that $\hat{\rho}'_j = \alpha_2 \Delta$. Then, $\rho'_j \leq
(1+\alpha_1)\hat{\rho}'_j = (1+\alpha_1)\alpha_2 \Delta$ whp., by the same
Chernoff bound as above. Now the upper bound of (\ref{eq:cut_sample_cor})
follows from $\rho_j \leq \rho'_j$, which holds for any instantiation of the
random sample $R$. The formal proof follows from Lemma~\ref{l:sampling}, with
$\beta = 1$, $d = 2$ and $q = 0$, and with $\Delta$ instead of $n^\delta$.
\qed\end{proofsketch}
We note that $\rho_j \geq 0$ and always assume that $\rho_j \leq \deg(j)$, since if $\rho_j$ satisfies (\ref{eq:cut_sample_cor}), $\min\{ \rho_j, \deg(j) \}$ also satisfies (\ref{eq:cut_sample_cor}). For all $\e_1, \e_2 > 0$, setting $\alpha_1 = \frac{\e_1}{1+\e_1}$ and $\alpha_2 = \e_2$ in Lemma~\ref{l:cut_sampling}, and taking the union bound over all vertices, we obtain that for $\gamma = \Theta(1/(\e^2_1 \e_2))$, with probability at least $1 - 2/n^2$, the following holds for all vertices $j \in V$:
\begin{equation}\label{eq:cut_sample}
 (1-\e_1)\rho_j - \e_2 \Delta \leq \hat{\rho}_j \leq
 (1+\e_1)\rho_j + \e_2 \Delta
\end{equation}
Therefore, with probability at least $1-2/n^2$, the optimal cut $\vec{x}^\ast$ is a feasible solution to (IP) with the estimations $\rho_j$ obtained by restricting $\vec{x}^\ast$ to the vertices in $R$. %So, from now on, we assume that $\vec{x}^\ast$ is a feasible solution to (IP) and to (LP).

\subsection{The Cut Value of Feasible Solutions}
\label{s:cut_linearization}

We next show that the objective value of any feasible solution $\vec{y}$ to (LP) is close to $p(\vec{y})$. Therefore, assuming that $\vec{x}^\ast$ is feasible, any good approximation to (IP) is a good approximation to the optimal cut. %The proof of the following can be found in the Appendix, Section~\ref{s:app:cut_approx}.
\begin{lemma}\label{l:cut_approx}
Let $\rho_1, \ldots, \rho_n$ be non-negative numbers and $\vec{y}$ be any feasible solution to (LP). Then,
\begin{equation}\label{eq:cut_approx}
 p(\vec{y}) \in \sum_{j \in V} y_j (\deg(j) - \rho_j) \pm 2(\e_1 + \e_2) m
\end{equation}
\end{lemma}
\begin{proof}
Using (\ref{eq:cut_decomp}) and the formulation of (LP), we obtain that:
\begin{align*}
 p(\vec{y}) = \sum_{j \in V} y_j \left(\deg(j) - \sum_{i \in N(j)} y_i\right) & \in
 \sum_{j \in V} y_j \left(\deg(j) - ((1\mp \e_1) \rho_j \mp \e_2 \Delta) \right)  \\
 &= \sum_{j \in V} y_j (\deg(j) - \rho_j) \pm \e_1 \sum_{j \in V} y_j \rho_j
    \pm \e_2 \Delta \sum_{j \in V} y_j \\
 &\in \sum_{j \in V} y_j (\deg(j) - \rho_j) \pm 2(\e_1 + \e_2) m
\end{align*}
The first inclusion holds because $\vec{y}$ is feasible for (LP) and thus, $\sum_{i \in N(j)} y_i \in (1\pm \e_1)\rho_j \pm \e_2\Delta$, for all $j$. The third inclusion holds because
\[ \sum_{j \in V} y_j \rho_j \leq \sum_{j \in V} \rho_j
   \leq \sum_{j \in V} \deg(j) = 2m\,,\]
since each $\rho_j$ is at most $\deg(j)$, and because $\Delta \sum_{j \in V} y_j \leq \Delta n = 2m$.
\qed\end{proof}

\subsection{Randomized Rounding of the Fractional Optimum}
\label{s:cut_rounding}

As a last step, we show how to round the fractional optimum $\vec{y}^\ast = (y^\ast_1, \ldots, y^\ast_n)$ of (LP) to an integral solution $\vec{z} = (z_1, \ldots, z_n)$ that almost satisfies the constraints of (IP).

To this end, we use randomized rounding, as in \cite{RT87}. In particular, we set independently each $z_j$ to $1$, with probability $y_j^\ast$, and to $0$, with probability $1-y_j^\ast$. By Chernoff bounds%
\footnote{\label{foot:chernoff}We use the following standard Chernoff bound (see e.g., \cite[Theorem~1.1]{DP09}): Let $Y_1, \ldots, Y_k$ independent random variables in $[0, 1]$ and let $Y = \sum_{j=1}^k Y_j$. Then for all $t > 0$, $\Prob[|Y - \Exp[Y]| > t] \leq 2\exp(-2t^2/k)$.},
we obtain that with probability at least $1 - 2/n^{8}$, for each vertex $j$,
\begin{equation}\label{eq:deviation}
 (1-\e_1)\rho_j - \e_2\Delta - 2\sqrt{\deg(j)\ln(n)} \leq
 \sum_{i \in N(j)} z_i  \leq
 (1+\e_1)\rho_j + \e_2\Delta + 2\sqrt{\deg(j)\ln(n)}
\end{equation}
Specifically, the inequality above follows from the Chernoff bound in footnote~\ref{foot:chernoff}, with $k = \deg(j)$ and $t = 2\sqrt{\deg(j)\ln(n)}$, since $\Exp[\sum_{i \in N(j)} z_j] = \sum_{i \in N(j)} y^\ast_j \in (1\pm\e_1)\rho_j \pm \e_2\Delta$. By the union bound, (\ref{eq:deviation}) is satisfied with probability at least $1 - 2/n^7$ for all vertices $j$.

By linearity of expectation, $\Exp[ \sum_{j \in V} z_j (\deg(j) - \rho_j) ] = \sum_{j \in V} y^\ast_j (\deg(j) - \rho_j)$. Moreover, since the probability that $\vec{z}$ does not satisfy (\ref{eq:deviation}) for some vertex $j$ is at most $2/n^7$ and since the objective value of (IP) is at most $n^2$, the expected value of a rounded solution $\vec{z}$ that satisfies (\ref{eq:deviation}) for all vertices $j$ is least $\sum_{j \in V} y^\ast_j (\deg(j) - \rho_j) - 1$ (assuming that $n \geq 2$). Using the method of conditional expectations, as in \cite{Rag88}, we can find in (deterministic) polynomial time an integral solution $\vec{z}$ that satisfies (\ref{eq:deviation}) for all vertices $j$ and has $\sum_{j \in V} z_j (\deg(j) - \rho_j) \geq \sum_{j \in V} y^\ast_j (\deg(j) - \rho_j) - 1$. Next, we sometimes abuse the notation and refer to such an integral solution $\vec{z}$ (computed deterministically) as the integral solution obtained from $\vec{y}^\ast$ by randomized rounding.

The following is similar to Lemma~\ref{l:cut_approx} and shows that the objective value $p(\vec{z})$ of the rounded solution $\vec{z}$ is close to the optimal value of (LP). %The proof can be found in the Appendix, Section~\ref{s:app:cut_approx2}.
\begin{lemma}\label{l:cut_approx2}
Let $\vec{y}^\ast$ be the optimal solution of (LP) and let $\vec{z}$ be the integral solution obtained from $\vec{y}^\ast$ by randomized rounding (and the method of conditional expectations). Then,
\begin{equation}\label{eq:cut_approx2}
 p(\vec{z}) \in \sum_{j \in V} y^\ast_j (\deg(j) - \rho_j) \pm 3(\e_1 + \e_2) m
\end{equation}
\end{lemma}
\begin{proof}
Using (\ref{eq:deviation}) and an argument similar to that in the proof of Lemma~\ref{l:cut_approx}, we obtain that:
\begin{align*}
 p(\vec{z}) & = \sum_{j \in V} z_j \left(\deg(j) - \sum_{i \in N(j)} z_i\right) \\
 & \in \sum_{j \in V} z_j \left(\deg(j) - \left((1\mp \e_1) \rho_j \mp \e_2 \Delta \mp 2\sqrt{\deg(j)\ln(n)}\right) \right)  \\
 &= \sum_{j \in V} z_j (\deg(j) - \rho_j) \pm \e_1 \sum_{j \in V} z_j \rho_j
    \pm \e_2 \Delta \sum_{j \in V} z_j \pm 2\sum_{j \in V} z_j \sqrt{\deg(j)\ln(n)}\\
 &\in \sum_{j \in V} z_j (\deg(j) - \rho_j) \pm (3\e_1 + 2\e_2) m \\
 &\in \sum_{j \in V} y^\ast_j (\deg(j) - \rho_j) \pm 3(\e_1 + \e_2) m \\
\end{align*}
The first inclusion holds because $\vec{z}$ satisfies (\ref{eq:deviation}) for all $j \in V$. For the third inclusion, we use that $\sum_{j \in V} z_j \rho_j \leq \sum_{j \in V} \deg(j) = 2m$, that $\Delta \sum_{i \in V} z_i \leq \Delta n = 2m$ and that by Jensen's inequality,
\[
  2 \sum_{j \in V} z_j \sqrt{\deg(j) \ln n} \leq
  \sum_{j \in V} \sqrt{4\,\deg(j) \ln n} \leq
  \sqrt{8 m n \ln n} \leq \e_1 m\,,
\]
assuming that $n$ and $m = \Omega(n^{1+\delta})$ are sufficiently large. For the last inclusion, we recall that $\sum_{j \in V} z_j (\deg(j) - \rho_j) \geq \sum_{j \in V} y^\ast_j (\deg(j) - \rho_j) - 1$ and assume that $m$ is sufficiently large.
\qed\end{proof}

\subsection{Putting Everything Together}
\label{s:together}

Therefore, for any $\eps > 0$, if $G$ is $\delta$-almost sparse and $\Delta = n^{\delta}$, the algorithm described in Section~\ref{s:cut_main}, with sample size $\Theta(n \ln n / (\eps^3 \Delta))$, computes estimations $\rho_j$ such that the optimal cut $\vec{x}^\ast$ is a feasible solution to (IP) whp. Hence, by the analysis above, the algorithm approximates the value of the optimal cut $p(\vec{x}^\ast)$ within an additive term of $O(\eps m)$. Specifically, setting $\e_1 = \e_2 = \eps/16$, the value of the cut $\vec{z}$ produced by the algorithm satisfies the following with probability at least $1-2/n^2$\,:
\[ p(\vec{z}) \geq \sum_{j \in V} y_j^\ast (\deg(j) - \rho_j) - 3 \eps m/8
              \geq \sum_{j \in V} x_j^\ast (\deg(j) - \rho_j) - 3 \eps m/8
              \geq p(\vec{x}^\ast) - \eps m / 2 \geq (1-\eps) p(\vec{x}^\ast)
\]
The first inequality follows from Lemma~\ref{l:cut_approx2}, the second inequality holds because $\vec{y}^\ast$ is the optimal solution to (LP) and $\vec{x}^\ast$ is feasible for (LP), the third inequality follows from Lemma~\ref{l:cut_approx} and the fourth inequality holds because the optimal cut has at least $m/2$ edges. 
\begin{theorem}\label{th:maxcut}
Let $G(V, E)$ be a $\delta$-almost sparse graph with $n$ vertices. Then, for any $\eps > 0$, we can compute, in time $2^{O(n^{1-\delta} \ln n/\eps^3)}$ and with probability at least $1-2/n^2$, a cut $\vec{z}$ of $G$ with value $p(\vec{z}) \geq (1-\eps)p(\vec{x}^\ast)$, where $\vec{x}^\ast$ is the optimal cut.
\end{theorem}
\section{Approximate Maximization of Smooth Polynomials}
\label{s:pip}

Generalizing the ideas applied to \MC, we arrive at the main algorithmic result
of the paper: an algorithm to approximately optimize $\beta$-smooth
$\delta$-bounded polynomials $p(\vec{x})$ of degree $d$ over all binary vectors
$\vec{x} \in \{0, 1\}^n$.  The intuition and the main ideas are quite similar
to those in Section~\ref{s:maxcut}, but the details are significantly more
involved because we are forced to recursively decompose degree $d$ polynomials
to eventually obtain a linear program.  
%In the following Section~\ref{s:app:pip}, 
In what follows, we take care of the technical details. %and prove the
%All technical details are given in the following theorem.

%\subsubsection{Approximate Maximization of Polynomials: The Proof of Theorem~\ref{th:pip_scheme}}
%\label{s:app:pip}

Next, we significantly generalize the ideas applied to \MC\ so that we approximately optimize $\beta$-smooth $\delta$-bounded polynomials $p(\vec{x})$ of degree $d$ over all binary vectors $\vec{x} \in \{0, 1\}^n$. The structure of this section deliberately parallels the structure of Section~\ref{s:maxcut}, so that the application to \MC\ can always serve as a reference for the intuition behind the generalization.

As in \cite{AKK99} (and as explained in Section~\ref{s:prelim}), we exploit the fact that any $n$-variate degree-$d$ $\beta$-smooth polynomial $p(\vec{x})$ can be decomposed into $n$ degree-$(d-1)$ $\beta$-smooth polynomials $p_j(\vec{x})$ such that $p(\vec{x}) = c + \sum_{j \in N} x_j p_j(\vec{x})$ (Proposition~\ref{pr:decomposition}).
For smooth polynomials of degree $d \geq 3$, we apply Proposition~\ref{pr:decomposition} recursively until we end up with smooth polynomials of degree $1$.
Specifically, using Proposition~\ref{pr:decomposition}, we further decompose each degree-$(d-1)$ $\beta$-smooth polynomial $p_{i_1}(\vec{x})$ into $n$ degree-$(d-2)$ $\beta$-smooth polynomials $p_{i_1 j}(\vec{x})$ such that
$p_{i_1}(\vec{x}) = c_{i_1} + \sum_{j \in N} x_j p_{i_1 j}(\vec{x})$, etc.
At the basis of the recursion, at depth $d-1$, we have $\beta$-smooth polynomials $p_{i_1\ldots i_{d-1}}(\vec{x})$ of degree $1$, one for each $(d-1)$-tuple of indices $(i_1, \ldots, i_{d-1}) \in N^{d-1}$. These polynomials are written as
\[ p_{i_1\ldots i_{d-1}}(\vec{x}) = c_{i_1\ldots i_{d-1}} +
                                    \sum_{j \in N} x_j c_{i_1\ldots i_{d-1} j}\,,
\]
where $c_{i_1\ldots i_{d-1} j}$ are constants (these are the coefficients of the corresponding degree-$d$ monomials in the expansion of $p(\vec{x})$). Due to $\beta$-smoothness, $|c_{i_1\ldots i_{d-1} j}| \leq \beta$ and $|c_{i_1\ldots i_{d-1}}| \leq \beta n$. Inductively, $\beta$-smoothness implies that each polynomial $p_{i_1\ldots i_{d-\ell}}(\vec{x})$ of degree $\ell \geq 1$ in this decomposition%
\footnote{This decomposition can be performed in a unique way if we insist that $i_1 < i_2 < \cdots < i_{d-1}$, but this is not important for our analysis.}
has $|p_{i_1\ldots i_{d-\ell}}(\vec{x})| \leq (\ell+1) \beta n^{\ell}$ for all binary vectors $\vec{x} \in \{0, 1\}^n$. Such a decomposition of $p(\vec{x})$ in $\beta$-smooth polynomials of degree $d-1, d-2, \ldots, 1$ can be computed recursively in time $O(n^d)$.

\subsection{Outline and General Approach}
\label{s:pip_outline}

As in Section~\ref{s:maxcut} (and as in \cite{AKK99}), we observe that if we have good estimations $\rho_{i_1\ldots i_{d-\ell}}$ of the values of each degree-$\ell$ polynomial $p_{i_1\ldots i_{d-\ell}}(\vec{x})$ at the optimal solution $\vec{x}^\ast$, for each level $\ell = 1, \ldots, d-1$ of the decomposition, then approximate maximization of $p(\vec{x})$ can be reduced to the solution of the following Integer Linear Program:
\begin{align}
\max \sum_{j \in N} y_j \rho_j \tag{$d$-IP}\\
\mathrm{s.t.}\ \ \ \ \ \ \
c_{i_1} + \sum_{j\in N} y_j \rho_{i_1 j} & \in
 \rho_{i_1} \pm \e_1 \rb_{i_1} \pm \e_2 n^{d-1+\delta} &
 \forall i_1 \in N \notag \\
c_{i_1i_2} + \sum_{j\in N} y_j \rho_{i_1 i_2 j} & \in
\rho_{i_1 i_2} \pm \e_1 \rb_{i_1 i_2} \pm \e_2 n^{d-2+\delta} &
\forall (i_1, i_2) \in N \times N \notag \\
\cdots \notag\\
c_{i_1 \ldots i_{d-\ell}} + \sum_{j\in N} y_j \rho_{i_1 \ldots i_{d-\ell} j} & \in
\rho_{i_1 \ldots i_{d-\ell}} \pm \e_1 \rb_{i_1 \ldots i_{d-\ell}}
\pm \e_2 n^{d-\ell+\delta} &
\forall (i_1, \ldots, i_{d-\ell}) \in N^{d-\ell} \notag \\
\cdots \notag\\
c_{i_1 \ldots i_{d-1}} + \sum_{j\in N} y_j c_{i_1 \ldots i_{d-1} j} & \in
\rho_{i_1 \ldots i_{d-1}} \pm \e_1 \rb_{i_1 \ldots i_{d-1}}\pm \e_2 n^{\delta} &
\forall (i_1, \ldots, i_{d-1}) \in N^{d-1} \notag \\
y_j & \in \{0, 1\} & \forall j \in N \notag
\end{align}
In ($d$-IP), we also use \emph{absolute value estimations} $\rb_{i_1 \ldots i_{d-\ell}}$. For each level $\ell \geq 1$ of the decomposition of $p(\vec{x})$ and each tuple $(i_1, \ldots, i_{d-\ell}) \in N^{d-\ell}$, we define the corresponding absolute value estimation as $\rb_{i_1 \ldots i_{d-\ell}} = \sum_{j \in N} |\rho_{i_1 \ldots i_{d-\ell}j}|$. Namely, each absolute value estimation $\rb_{i_1 \ldots i_{d-\ell}}$ at level $\ell$ is the sum of the absolute values of the estimations $\rho_{i_1 \ldots i_{d-\ell}j}$ at level $\ell-1$.
The reason that we use absolute value estimations and set the lhs/rhs of the constraints to $\rho_{i_1 \ldots i_{d-\ell}} \pm \e_1 \rb_{i_1 \ldots i_{d-\ell}}$, instead of simply to $(1\pm\e_1)\rho_{i_1 \ldots i_{d-\ell}}$, is that we want to consider linear combinations of positive and negative estimations $\rho_{i_1 \ldots i_{d-\ell}}$ in a uniform way.

%In ($d$-IP), we use absolute value estimations $\rb_{i_1 \ldots i_{d-\ell}}$, which are defined recursively as follows: For level $\ell = 1$ and each tuple $(i_1, \ldots, i_{d-1}) \in N^{d-1}$, we let $\rb_{i_1 \ldots i_{d-1}} = \sum_{j \in N} |c_{i_1 \ldots i_{d-1}j}|$. For each level $\ell \geq 2$ of the decomposition of $p(\vec{x})$ and each tuple $(i_1, \ldots, i_{d-\ell}) \in N^{d-\ell}$, we let $\rb_{i_1 \ldots i_{d-\ell}} = \sum_{j \in N} (|\rho_{i_1 \ldots i_{d-\ell}j}| + \rb_{i_1 \ldots i_{d-\ell}j})$. Intuitively, each absolute value estimation $\rb_{i_1 \ldots i_{d-\ell}}$ is equal to the sum of the absolute values of all estimations below the root of the decomposition tree of $p_{i_1 \ldots i_{d-\ell}}(\vec{x})$.

Similarly to Section~\ref{s:maxcut}, the estimations $\rho_{i_1 \ldots i_{d-\ell}}$ (and $\rb_{i_1 \ldots i_{d-\ell}}$) are computed (by exhaustive sampling) and the constants $\e_1, \e_2 > 0$ are calculated so that the optimal solution $\vec{x}^\ast$ is a feasible solution to ($d$-IP). In the following, we let $\vec{\rho}$ denote the sequence of estimations $\rho_{i_1 \ldots i_{d-\ell}}$, for all levels $\ell$ and all tuples $(i_1, \ldots, i_{d-\ell}) \in N^{d-\ell}$, that we use to formulate ($d$-IP). The absolute value estimations $\rb_{i_1 \ldots i_{d-\ell}}$ can be easily computed from $\vec{\rho}$. We let ($d$-LP) denote the Linear Programming relaxation of ($d$-IP), where each $y_j \in [0, 1]$, let $\vec{x}^\ast$ denote the binary vector that maximizes $p(\vec{x})$, and let $\vec{y}^\ast \in [0,1]^n$ denote the fractional optimal solution of ($d$-LP).

As in Section~\ref{s:maxcut}, the approach is based on the facts that (i) for all constants $\e_1, \e_2 > 0$, we can compute estimations $\vec{\rho}$, by exhaustive sampling, so that $\vec{x}^\ast$ is a feasible solution to ($d$-IP) with high probability (see Lemma~\ref{l:sampling} and Lemma~\ref{l:sampling_gen}); and that (ii) the objective value of any feasible solution $\vec{y}$ to ($d$-LP) is close to $p(\vec{y})$ (see Lemma~\ref{l:approx} and Lemma~\ref{l:approx_gen}). Based on these observations, the general description of the approximation algorithm is essentially identical to the three steps described in Section~\ref{s:cut_main} and the reasoning behind the approximation guarantee is that of (\ref{eq:cut_est}).

\subsection{Obtaining Estimations by Exhausting Sampling}
\label{s:pip_sampling}

We first show how to use exhaustive sampling and obtain an estimation $\rho_{i_1\ldots i_{d-\ell}}$ of the value at the optimal solution $\vec{x}^\ast$ of each degree-$\ell$ polynomial $p_{i_1\ldots i_{d-\ell}}(\vec{x})$ in the decomposition of $p(\vec{x})$.

As in Section~\ref{s:cut_sampling}, we take a sample $R$ from $N$, uniformly at random and with replacement. The sample size is $r = \Theta(n^{1-\delta} \ln n)$. We try exhaustively all $0/1$ assignments to the variables in $R$, which can performed in time $2^r = 2^{O(n^{1-\delta}\ln n)}$.

\def\Est{\mathrm{Estimate}}
\begin{algorithm}[t]
\caption{\label{alg:estimate}Recursive estimation procedure $\Est(p_{i_1\ldots i_{d-\ell}}(\vec{x}), \ell, R, \vec{s})$}
\begin{algorithmic}\normalsize
    \Require $n$-variate degree-$\ell$ polynomial $p_{i_1\ldots i_{d-\ell}}(\vec{x})$, $R \subseteq N$ and a value $s_j \in \{0,1\}$ for each $j \in R$
    \Ensure Estimation $\rho_{i_1\ldots i_{d-\ell}}$ of $p_{i_1\ldots i_{d-\ell}}(\overline{\vec{s}})$, where $\overline{\vec{s}}_R = \vec{s}$

    \medskip\If{$\ell = 0$} \Return $c_{i_1\ldots i_{d}}$
    \ \ \ /* $p_{i_1\ldots i_{d}}(\vec{x})$ is equal to the constant $c_{i_1\ldots i_{d}}$ */ \EndIf
    \State compute decomposition
    $p_{i_1\ldots i_{d-\ell}}(\vec{x}) =
      c_{i_1\ldots i_{d-\ell}} + \sum_{j \in N} x_j p_{i_1\ldots i_{d-\ell}j}(\vec{x})$
    \For{all $j \in N$}
        \State $\rho_{i_1\ldots i_{d-\ell}j} \leftarrow \Est(p_{i_1\ldots i_{d-\ell}j}(\vec{x}), \ell-1, R, \vec{s})$
    \EndFor
    \State $\rho_{i_1\ldots i_{d-\ell}} \leftarrow c_{i_1\ldots i_{d-\ell}} + \frac{|N|}{|R|} \sum_{j \in R} s_j \rho_{i_1\ldots i_{d-\ell}j}$\\
    \Return $\rho_{i_1\ldots i_{d-\ell}}$
\end{algorithmic}\end{algorithm}

For each assignment, described by a $0/1$ vector $\vec{s}$ restricted to $R$,
we compute the corresponding estimations recursively, as described in Algorithm~\ref{alg:estimate}. Specifically, for the basis level $\ell = 0$ and each $d$-tuple $(i_1, \ldots, i_d) \in N^d$ of indices, the corresponding estimation is the coefficient $c_{i_1\ldots i_d}$ of the monomial $x_{i_1}\cdots x_{i_d}$ in the expansion of $p(\vec{x})$.
For each level $\ell$, $1 \leq \ell \leq d-1$, and each $(d-\ell)$-tuple $(i_1, \ldots, i_{d-\ell}) \in N^{d-\ell}$, given the level-$(\ell-1)$ estimations $\rho_{i_1\ldots i_{d-\ell} j}$ of $p_{i_1\ldots i_{d-\ell} j}(\overline{\vec{s}})$, for all $j \in N$, we compute the level-$\ell$ estimation $\rho_{i_1\ldots i_{d-\ell}}$ of $p_{i_1\ldots i_{d-\ell}}(\overline{\vec{s}})$ from $\vec{s}$ as follows:
\begin{equation}\label{eq:estimation}
   \rho_{i_1\ldots i_{d-\ell}} = c_{i_1\ldots i_{d-\ell}} +
   \frac{n}{r} \sum_{j \in R} s_j \rho_{i_1\cdots i_{d-\ell} j}
\end{equation}
In Algorithm~\ref{alg:estimate}, $\overline{\vec{s}}$ is any vector in $\{ 0, 1 \}^n$ that agrees with $\vec{s}$ on the variables of $R$. Given the estimations $\rho_{i_1 \ldots i_{d-\ell}j}$, for all $j \in N$, we can also compute the absolute value estimations $\rb_{i_1 \ldots i_{d-\ell}}$ at level $\ell$. Due to the $\beta$-smoothness property of $p(\vec{x})$, we have that $|c_{i_1\ldots i_{d-\ell}}| \leq \beta n^\ell$, for all levels $\ell \geq 0$. Moreover, we assume that $0 \leq \rb_{i_1\ldots i_{d-\ell}} \leq \ell\beta n^{\ell}$ and $|\rho_{i_1\ldots i_{d-\ell}}| \leq (\ell+1)\beta n^{\ell}$, for all levels $\ell \geq 1$. This assumption is wlog. because due to $\beta$-smoothness, any binary vector $\vec{x}$ is feasible for ($d$-IP) with such values for the estimations $\rho_{i_1\ldots i_{d-\ell}}$ and the absolute value estimations $\rb_{i_1\ldots i_{d-\ell}}$\,.
\begin{remark}
For simplicity, we state Algorithm~\ref{alg:estimate} so that it computes, from $\vec{s}$, an estimation $\rho_{i_1\ldots i_{d-\ell}}$ of the value of a given degree-$\ell$ polynomial $p_{i_1\ldots i_{d-\ell}}(\vec{x})$ at $\overline{\vec{s}}$. So, we need to apply Algorithm~\ref{alg:estimate} $O(n^{d-1})$ times, one for each polynomial that arises in the recursive decomposition, with the same sample $R$ and the same assignment $\vec{s}$. We can easily modify Algorithm~\ref{alg:estimate} so that a single call $\Est(p(\vec{x}), d, R, \vec{s})$ computes the estimations of all the polynomials that arise in the recursive decomposition of $p(\vec{x})$. Thus, we save a factor of $d$ on the running time. The running time of the simple version is $O(dn^d)$, while the running time of the modified version is $O(n^d)$.
\end{remark}

\subsection{Sampling Lemma}
\label{s:sampling}

We use the next lemma to show that if $\vec{s} = \vec{x}^\ast_R$, the estimations $\rho_{i_1\ldots i_{d-\ell}}$ computed by Algorithm~\ref{alg:estimate} are close to $c_{i_1\ldots i_{d-\ell}} + \sum_{j \in N} x^\ast_j \rho_{i_1\ldots i_{d-\ell} j}$ with high probability.
%
%(and subsequently close to $p_{i_1\ldots i_{d-\ell}}(\vec{x}^\ast)$, see also Lemma~\ref{l:feasible}).
\begin{lemma}\label{l:sampling}
Let $\vec{x}$ be any binary vector and let $( \rho_j )_{j \in N}$ be any sequence such that for some integer $q \geq 0$ and some constant $\beta \geq 1$, $\rho_j \in [0, (q+1)\beta n^q]$, for all $j \in N$. For all integers $d \geq 1$ and for all $\alpha_1, \alpha_2 > 0$, we let $\gamma = \Theta(d q \beta/(\alpha_1^2 \alpha_2))$ and let $R$ be a multiset of $r = \gamma n^{1-\delta} \ln n$ indices chosen uniformly at random with replacement from $N$, where $\delta \in (0, 1]$ is any constant. If $\rho = (n / r) \sum_{j \in R} \rho_{j} x_j$ and $\hat{\rho} = \sum_{j \in N} \rho_{j} x_j$, with probability at least $1 - 2/n^{d+1}$,
\begin{equation}\label{eq:pip_sample}
 (1-\alpha_1)\hat{\rho} - (1-\alpha_1)\alpha_2 n^{q+\delta} \leq \rho \leq
 (1+\alpha_1)\hat{\rho} + (1+\alpha_1)\alpha_2 n^{q+\delta}
\end{equation}
\end{lemma}
\begin{proof}
To provide some intuition, we observe that if $\hat{\rho} = \Omega(n^{q+\delta})$, we have $\Omega(n^\delta)$ values $\rho_j = \Theta(n^q)$. These values are well-represented in the random sample $R$, with high probability, since the size of the sample is $\Theta(n^{1-\delta} \ln n)$. Therefore, $|\hat{\rho} - \rho| \leq \alpha_1\hat{\rho}$, with high probability, by standard Chernoff bounds. If $\hat{\rho} = o(n^{q+\delta})$, the lower bound in (\ref{eq:pip_sample}) becomes trivial, since it is non-positive, while $\rho \geq 0$. As for the upper bound, we increase the coefficients $\rho_j$ to $\rho'_j \in [0, (q+1)\beta n^q]$, so that $\hat{\rho}' = \alpha_2 n^{q+\delta}$. Then, $\rho' \leq (1+\alpha_1)\hat{\rho}' = (1+\alpha_1)\alpha_2 n^{q+\delta}$, with high probability, by the same Chernoff bound as above. Now the upper bound of (\ref{eq:pip_sample}) follows from $\rho \leq \rho'$, which holds for any instantiation of the random sample $R$.

We proceed to formalize the idea above. For simplicity of notation, we let $B = (q+1)\beta n^q$ and $a_2 = \alpha_2/((q+1)\beta)$ throughout the proof. For each sample $l$, $l = 1, \ldots, r$, we let $X_l$ be a random variable distributed in $[0, 1]$. For each index $j$, if the $l$-th sample is $j$, $X_l$ becomes $\rho_{j} / B$, if $x_j = 1$, and becomes $0$, otherwise. Therefore, $\Exp[X_l] = \hat{\rho} / (B n)$. We let $X = \sum_{l = 1}^r X_l$. Namely, $X$ is the sum of $r$ independent random variables identically distributed in $[0, 1]$. Using that $r = \gamma n^{1-\delta} \ln n$, we have that $\Exp[X] = \gamma \hat{\rho} \ln n / (B n^\delta)$ and that $\rho = B n X/r = B n^\delta X / (\gamma \ln n)$.

We distinguish between the case where $\hat{\rho} \geq a_2 B n^{\delta}$ and the case where $\hat{\rho} < a_2 B n^{\delta}$.
We start with the case where $\hat{\rho} \geq a_2 B n^{\delta}$. Then, by Chernoff bounds%
\footnote{\label{foot:chernoff2}We use the following bound (see e.g., \cite[Theorem~1.1]{DP09}): Let $Y_1, \ldots, Y_k$ be independent random variables identically distributed in $[0, 1]$ and let $Y = \sum_{j=1}^k Y_j$. Then for all $\e \in (0, 1)$, $\Prob[|Y - \Exp[Y]| > \e\, \Exp[Y]] \leq 2\exp(-\e^2\,\Exp[Y]/3)$.},
\begin{eqnarray*}
 \Prob[|X - \Exp[X]| > \alpha_1 \Exp[X]] & \leq &
 2\exp\!\left(-\frac{\alpha_1^2 \gamma \hat{\rho} \ln n}{3 B n^{\delta} }\right) \\
 & \leq & 2\exp(-\alpha_1^2 a_2 \gamma \ln n / 3) \leq 2/n^{d+1}
\end{eqnarray*}
For the second inequality, we use that $\hat{\rho} \geq a_2 B n^{\delta}$. For the last inequality, we use that $\gamma \geq 3(d+1)/(\alpha_1^2 a_2) = 3(d+1)(q+1)\beta/(\alpha_1^2 \alpha_2)$, since $a_2 = \alpha_2/((q+1)\beta)$. Therefore, with probability at least $1 - 2/n^{d+1}$,
\[
 (1-\alpha_1) \frac{\gamma \hat{\rho} \ln n}{B n^\delta} \leq X \leq
 (1+\alpha_1) \frac{\gamma \hat{\rho} \ln n}{B n^\delta}
\]
Multiplying everything by $B n / r = B n^\delta /(\gamma \ln n)$, we have that with probability at least $1-2/n^{d+1}$, $(1-\alpha_1) \hat{\rho} \leq \rho \leq (1+\alpha_1) \hat{\rho}$, which clearly implies (\ref{eq:pip_sample}).

We proceed to the case where $\hat{\rho} < a_2 B n^{\delta}$. Then, $(1-\alpha_1)\hat{\rho} < (1-\alpha_1) a_2 B n^{\delta} = (1-\alpha_1) \alpha_2 n^{q+\delta}$. Therefore, since $\rho \geq 0$, because $\rho_j \geq 0$, for all $j \in N$, the lower bound of (\ref{eq:pip_sample}) on $\rho$ is trivial.
For the upper bound, we show that with probability at least $1-1/n^{d+1}$, $\rho \leq (1+\alpha_1) a_2 B n^{\delta} = (1+\alpha_1)\alpha_2n^{q+\delta}$. To this end, we consider a sequence $(\rho'_j)_{j \in N}$ so that $\rho_j \leq \rho'_j \leq (q+1)\beta n^q$, for all $j \in N$, and
\( \hat{\rho}' = \sum_{j \in N} \rho'_{j} x_j = a_2 B n^{q+\delta} \).
We can obtain such a sequence by increasing an appropriate subset of $\rho_j$ up to $(q+1)\beta n^q$ (if $\vec{x}$ does not contain enough $1$'s, we may also change some $x_j$ from $0$ to $1$).
For the new sequence, we let $\rho' = (n / r) \sum_{j \in R} \rho'_{j} x_j$ and observe that $\rho \leq \rho'$, for any instantiation of the random sample $R$.
Therefore,
\[ \Prob[\rho > (1+\alpha_1)\alpha_2n^{q+\delta}] \leq
   \Prob[\rho' > (1+\alpha_1)\hat{\rho}']\,,
\]
where we use that $\hat{\rho}' = a_2 B n^\delta = \alpha_2n^{q+\delta}$.
By the choice of $\hat{\rho}'$, we can apply the same Chernoff bound as above and obtain that $\Prob[\rho' > (1+\alpha_1)\hat{\rho}'] \leq 1/n^{d+1}$.
\qed\end{proof}
Lemma~\ref{l:sampling} is enough for \MC\ and graph optimization problems, where the estimations $\rho_{i_1\ldots i_{d-\ell} j}$ are non-negative. For arbitrary smooth polynomials however, the estimations $\rho_{i_1\ldots i_{d-\ell} j}$ may also be negative. So, we need a generalization of Lemma~\ref{l:sampling} that deals with both positive and negative estimations. To this end, given a sequence of estimations $( \rho_j )_{j \in N}$, with $\rho_j \in [-(q+1)\beta n^q, (q+1)\beta n^q]$, we let $\rho^+_j = \max\{\rho_j, 0\}$ and $\rho^-_j = \min\{ \rho_j, 0\}$, for all $j \in N$. Namely, $\rho^+_j$ (resp. $\rho^-_j$) is equal to $\rho_j$, if $\rho_j$ is positive (resp. negative), and $0$, otherwise. Moreover, we let 
\[ \rho^+ = (n / r) \sum_{j \in R} \rho^+_{j} x_j\,,\ \  
   \hat{\rho}^+ = \sum_{j \in N} \rho^+_{j} x_j\,,\ \ 
   \rho^- = (n / r) \sum_{j \in R} \rho^-_{j} x_j \mbox{\ \ and\ \ }   
   \hat{\rho}^- = \sum_{j \in N} \rho^-_{j} x_j 
\]
Applying Lemma~\ref{l:sampling} once for positive estimations and once for negative estimations (with the absolute values of $\rho_j^-$, $\rho^-$ and $\hat{\rho}^-$, instead), we obtain that with probability at least $1 - 4/n^{d+1}$, the following inequalities hold:
\begin{eqnarray*}
  (1-\alpha_1)\hat{\rho}^+ - (1-\alpha_1)\alpha_2 n^{q+\delta} \leq & \rho^+ & \leq
  (1+\alpha_1)\hat{\rho}^+ + (1+\alpha_1)\alpha_2 n^{q+\delta} \\
  (1+\alpha_1)\hat{\rho}^- - (1+\alpha_1)\alpha_2 n^{q+\delta} \leq & \rho^- & \leq
  (1-\alpha_1)\hat{\rho}^- + (1-\alpha_1)\alpha_2 n^{q+\delta}
\end{eqnarray*}
Using that $\rho = \rho^+ + \rho^-$ and that $\hat{\rho} = \hat{\rho}^+ + \hat{\rho}^-$, we obtain the following generalization of Lemma~\ref{l:sampling}.
\begin{lemma}[Sampling Lemma]\label{l:sampling_gen}
Let $\vec{x} \in \{0, 1\}^n$ and let $( \rho_j )_{j \in N}$ be any sequence such that for some integer $q \geq 0$ and some constant $\beta \geq 1$, $|\rho_j| \leq  (q+1)\beta n^q$, for all $j \in N$. For all integers $d \geq 1$ and for all $\alpha_1, \alpha_2 > 0$, we let $\gamma = \Theta(d q \beta/(\alpha_1^2 \alpha_2))$ and let $R$ be a multiset of $r = \gamma n^{1-\delta} \ln n$ indices chosen uniformly at random with replacement from $N$, where $\delta \in (0, 1]$ is any constant. If $\rho = (n / r) \sum_{j \in R} \rho_{j} x_j$, $\hat{\rho} = \sum_{j \in N} \rho_{j} x_j$ and $\rb = \sum_{j \in N} |\rho_j|$, with probability at least $1 - 4/n^{d+1}$,
\begin{equation}\label{eq:pip_sample_gen}
 \hat{\rho} - \alpha_1 \rb - 2\alpha_2 n^{q+\delta} \leq \rho \leq
 \hat{\rho} + \alpha_1 \rb + 2\alpha_2 n^{q+\delta}
\end{equation}
\end{lemma}
For all constants $\e_1, \e_2 > 0$ and all constants $c$, we use Lemma~\ref{l:sampling_gen} with $\alpha_1 = \e_1$ and $\alpha_2 = \e_2/2$ and obtain that for $\gamma = \Theta(d q \beta /(\e^2_1 \e_2))$, with probability at least $1 - 4/n^{d+1}$, the following holds for any binary vector $\vec{x}$ and any sequence of estimations $( \rho_j )_{j \in N}$ produced by Algorithm~\ref{alg:estimate} with $\vec{s} = \vec{x}_R$ (note that in Algorithm~\ref{alg:estimate}, the additive constant $c$ is included in the estimation $\rho$ when its value is computed from the estimations $\rho_j$).
\begin{equation}\label{eq:pip_sample2}
 \overbrace{c+\frac{n}{r}\sum_{j \in R} \rho_j x_j}^{\rho}
 - \e_1 \overbrace{\sum_{j \in N} |\rho_j|}^{\rb} - \e_2 n^{q+\delta} \leq
 c + \sum_{j \in N} x_j \rho_j \leq
 \overbrace{c+\frac{n}{r}\sum_{j \in R} \rho_j x_j}^{\rho} 
 + \e_1 \overbrace{\sum_{j \in N} |\rho_j|}^{\rb} + \e_2 n^{q+\delta}
\end{equation}
%Whenever we apply (\ref{eq:pip_sample2}), we assume that $c - \rb \leq \rho \leq c + \rb$, with the more general definition of absolute value estimations provided in ($d$-IP). This assumption is wlog. because if $\rho$ satisfies (\ref{eq:pip_sample2}), $\min\{ \rho, c+\rb \}$ and $\max\{\rho, c - \rb\}$ also satisfy (\ref{eq:pip_sample2}) (note that this condition can be easily enforced in Algorithm~\ref{alg:estimate}).
%
Now, let us consider ($d$-IP) with the estimations computed by Algorithm~\ref{alg:estimate} with $\vec{s} = \vec{x}^\ast_R$ (i.e., with the optimal assignment for the variables in the random sample $R$). Then, using (\ref{eq:pip_sample2}) and taking the union bound over all constraints, which are at most $2n^{d-1}$, we obtain that with probability at least $1-8/n^2$, the optimal solution $\vec{x}^\ast$ is a feasible solution to ($d$-IP). So, from now on, we condition on the high probability event that $\vec{x}^\ast$ is a feasible solution to ($d$-IP) and to ($d$-LP).

\subsection{The Value of Feasible Solutions to ($d$-LP)}
\label{s:pip_value}

From now on, we focus on estimations $\vec{\rho}$ produced by $\Est(p(\vec{x}), d, R, \vec{s})$, where $R$ is a random sample from $N$ and $\vec{s} = \vec{x}^\ast_R$, and the corresponding programs ($d$-IP) and ($d$-LP). The analysis in Section~\ref{s:pip_sampling} implies that $\vec{x}^\ast$ is a feasible solution to ($d$-IP) (and to ($d$-LP)), with high probability.

We next show that for any feasible solution $\vec{y}$ of ($d$-LP) and any polynomial $q(\vec{x})$ in the decomposition of $p(\vec{x})$, the value of $q(\vec{y})$ is close to the value of $c + \sum_j y_j \rho_j$ in the constraint of ($d$-LP) corresponding to $q$. Applying Lemma~\ref{l:approx}, we show below (see Lemma~\ref{l:approx_gen}) that $p(\vec{y})$ is close to $c+\sum_{j \in N} y_j \rho_j$, i.e., to the objective value of $\vec{y}$ in ($d$-LP) and ($d$-IP), for any feasible solution $\vec{y}$.

To state and prove the following lemma, we introduce \emph{cumulative absolute value estimations} $\tb_{i_1 \ldots i_{d-\ell}}$\,, defined recursively as follows:
For level $\ell = 1$ and each tuple $(i_1, \ldots, i_{d-1}) \in N^{d-1}$, we let $\tb_{i_1 \ldots i_{d-1}} = \rb_{i_1 \ldots i_{d-1}} = \sum_{j \in N} |c_{i_1 \ldots i_{d-1}j}|$.
For each level $\ell \geq 2$ of the decomposition of $p(\vec{x})$ and each tuple $(i_1, \ldots, i_{d-\ell}) \in N^{d-\ell}$, we let $\tb_{i_1 \ldots i_{d-\ell}} = \rb_{i_1 \ldots i_{d-\ell}} + \sum_{j \in N} \tb_{i_1 \ldots i_{d-\ell}j}$. Namely, each cumulative absolute value estimation $\tb_{i_1 \ldots i_{d-\ell}}$ is equal to the sum of all absolute value estimations that appear below the root of the decomposition tree of $p_{i_1 \ldots i_{d-\ell}}(\vec{x})$.
\begin{lemma}\label{l:approx}
Let $q(\vec{x})$ be any $\ell$-degree polynomial appearing in the decomposition of $p(\vec{x})$, let $q(\vec{x}) = c+\sum_{j \in N} x_j q_j(\vec{x})$ be the decomposition of $q(\vec{x})$, let $\rho$ and $\{ \rho_j \}_{j \in N}$ be the estimations of $q$ and $\{ q_j \}_{j \in N}$ produced by Algorithm~\ref{alg:estimate} and used in ($d$-LP), and let $\tb$ and $\{ \tb_j \}_{j \in N}$ be the corresponding cumulative absolute value estimations. Then, for any feasible solution $\vec{y}$ of ($d$-LP)
\begin{equation}\label{eq:approx}
\rho - \e_1 \tb - \ell \e_2 n^{\ell - 1+\delta} \leq q(\vec{y}) \leq
\rho + \e_1 \tb + \ell \e_2 n^{\ell - 1+\delta}
\end{equation}
\end{lemma}
\begin{proof}
The proof is by induction on the degree $\ell$. The basis, for $\ell=1$, is trivial, because in the decomposition of $q(\vec{x})$, each $q_j(\vec{x})$ is a constant $c_j$. Therefore, Algorithm~\ref{alg:estimate} outputs $\rho_j = c_j$ and
\[ q(\vec{y}) = c + \sum_{j \in N} y_j q_j(\vec{x})
 = c + \sum_{j \in N} y_j c_j
 \in \rho \pm \e_1 \tb \pm \e_2 n^{\delta}\,,
 \]
where the inclusion follows from the feasibility of $\vec{y}$ for ($d$-LP). We also use that at level $\ell = 1$, $\tb = \rb$ (i.e., cumulative absolute value estimations and absolute value estimations are identical).

We inductively assume that (\ref{eq:approx}) is true for all degree-$(\ell-1)$ polynomials $q_j(\vec{x})$ that appear in the decomposition of $q(\vec{x})$ and establish the lemma for $q(\vec{x}) = c + \sum_{j \in N} x_j q_j(\vec{x})$. We have that:
\begin{align*}
 q(\vec{y}) = c + \sum_{j \in N} y_j q_j(\vec{y}) & \in
 c + \sum_{j \in N} y_j \left( \rho_j \pm \e_1 \tb_j
                                \pm (\ell-1) \e_2 n^{\ell-2+\delta} \right)\\
 &= \left(c + \sum_{j \in N} y_j \rho_j \right)
      \pm \e_1 \sum_{j \in N} y_j \tb_j
      \pm (\ell-1) \e_2 \sum_{j \in N} y_j n^{\ell-2+\delta} \\
 &\in \left(\rho \pm \e_1 \rb \pm \e_2 n^{\ell-1+\delta}\right)
      \pm \e_1 \sum_{j \in N} \tb_j \pm (\ell-1) \e_2 n^{\ell-1+\delta}\\
 &\in \rho \pm \e_1 \tb \pm \ell \e_2 n^{\ell-1+\delta}
\end{align*}
The first inclusion holds by the induction hypothesis. The second inclusion holds because (i) $\vec{y}$ is a feasible solution to ($d$-LP) and thus, $c + \sum_{j \in N} y_j \rho_j$ satisfies the corresponding constraint; (ii) $\sum_{j \in N} y_j \tb_j \leq \sum_{j \in N} \tb_j$; and (iii) $\sum_{j \in N} y_j \leq n$. The last inclusion holds because $\tb = \rb + \sum_{j \in N} \tb_j$, by the definition of cumulative absolute value estimations.
\qed\end{proof}
Using Lemma~\ref{l:approx} and the notion of cumulative absolute value estimations, we next show that $p(\vec{y})$ is close to $c+\sum_{j \in N} y_j \rho_j$, for any feasible solution $\vec{y}$.
\begin{lemma}\label{l:approx_gen}
Let $p(\vec{x}) = c+\sum_{j \in N} x_j p_j(\vec{x})$ be the decomposition of $p(\vec{x})$, let $\{ \rho_j \}_{j \in N}$ be the estimations of $\{ p_j \}_{j \in N}$ produced by Algorithm~\ref{alg:estimate} and used in ($d$-LP), and let $\{ \tb_j \}_{j \in N}$ be the corresponding cumulative absolute value estimations. Then, for any feasible solution $\vec{y}$ of ($d$-LP)
\begin{equation}\label{eq:approx_gen}
 p(\vec{y}) \in
 c+\sum_{j \in N} y_j \rho_j \pm \e_1 \sum_{j \in N} \tb_j \pm (d-1)\e_2 n^{d-1+\delta}
\end{equation}
\end{lemma}
\begin{proof}
By Lemma~\ref{l:approx}, for any polynomial $p_j$, $p_j(\vec{y}) \in \rho_j \pm \e_1 \tb_j \pm (d-1) \e_2 n^{d-2+\delta}$. Therefore,
\begin{align*}
 p(\vec{y}) = c + \sum_{j \in N} y_j p_j(\vec{y}) & \in
 c + \sum_{j \in N} y_j \left( \rho_j \pm \e_1 \tb_j
                                \pm (d-1) \e_2 n^{d-2+\delta} \right)\\
 &= c + \sum_{j \in N} y_j \rho_j
      \pm \e_1 \sum_{j \in N} y_j \tb_j
      \pm (d-1) \e_2 \sum_{j \in N} y_j n^{d-2+\delta} \\
 &\in c + \sum_{j \in N} y_j \rho_j
     \pm \e_1 \sum_{j \in N} \tb_j
     \pm (d-1) \e_2 n^{d-1+\delta}
\end{align*}
The second inclusion holds because $y_j \in [0,1]$ and $\sum_{j \in N} y_j \leq n$.
\qed\end{proof}

\subsection{Randomized Rounding of the Fractional Optimum}
\label{s:pip_rounding}

The last step is to round the fractional optimum $\vec{y}^\ast = (y^\ast_1, \ldots, y^\ast_n)$ of ($d$-LP) to an integral solution $\vec{z} = (z_1, \ldots, z_n)$ that almost satisfies the constraints of ($d$-IP) and has an expected objective value for ($d$-IP) very close to the objective value of $\vec{y}^\ast$.

To this end, we use randomized rounding, as in \cite{RT87}. In particular, we set independently each $z_j$ to $1$, with probability $y_j^\ast$, and to $0$, with probability $1-y_j^\ast$. The analysis is based on the following lemma, whose proof is similar to the proof of Lemma~\ref{l:sampling}.
\begin{lemma}\label{l:rounding}
Let $\vec{y} \in [0, 1]^n$ be any fractional vector and let $\vec{z} \in \{0, 1\}^n$ be an integral vector obtained from $\vec{y}$ by randomized rounding. Also, let $( \rho_j )_{j \in N}$ be any sequence such that for some integer $q \geq 0$ and some constant $\beta \geq 1$, $\rho_j \in [0, (q+1)\beta n^q]$, for all $j \in N$. For all integers $k \geq 1$ and for all constants $\alpha, \delta > 0$ (and assuming that $n$ is sufficiently large), if $\rho = \sum_{j \in N} \rho_{j} z_j$ and $\hat{\rho} = \sum_{j \in N} \rho_{j} y_j$, with probability at least $1 - 2/n^{k+1}$,
\begin{equation}\label{eq:rounding}
 (1-\alpha)\hat{\rho} - (1-\alpha)\alpha n^{q+\delta} \leq \rho \leq
 (1+\alpha)\hat{\rho} + (1+\alpha)\alpha n^{q+\delta}
\end{equation}
\end{lemma}
\begin{proof}
We first note that $\Exp[\rho] = \hat{\rho}$. If $\hat{\rho} = \Omega(n^{q} \ln n)$, then $|\rho - \hat{\rho}| \leq \alpha \hat{\rho}$, with high probability, by standard Chernoff bounds. If $\hat{\rho} = o(n^{q} \ln n)$, the lower bound in (\ref{eq:rounding}) becomes trivial, because $\rho \geq 0$ and $o(n^{q} \ln n) < \alpha n^{q+\delta}$, if $n$ is sufficiently large. As for the upper bound, we increase the coefficients $\rho_j$ to $\rho'_j \in [0, (q+1)\beta n^q]$, so that $\hat{\rho}' = \Theta(n^{q} \ln n)$. Then, the upper bound is shown as in the second part of the proof of Lemma~\ref{l:sampling}.

We proceed to the formal proof. For simplicity of notation, we let $B = (q+1)\beta n^q$ throughout the proof. For $j = 1, \ldots, n$, we let $X_j = z_j \rho_j / B$ be a random variable distributed in $[0, 1]$. Each $X_j$ independently takes the value $\rho_{j} / B$, with probability $y_j$, and $0$, otherwise. We let $X = \sum_{j = 1}^n X_j$ be the sum of these independent random variables. Then, $\Exp[X] = \hat{\rho} / B$ and $X = \sum_{j \in N} z_j \rho_j / B = \rho/B$.

As in Lemma~\ref{l:sampling}, we distinguish between the case where $\hat{\rho} \geq 3(k+1)B\ln n/\alpha^2$ and the case where $\hat{\rho} < 3(k+1)B\ln n/\alpha^2$.
We start with the case where $\hat{\rho} \geq 3(k+1)B\ln n/\alpha^2$. Then, by Chernoff bounds (we use the bound in footnote~\ref{foot:chernoff2}),
\[
 \Prob[|X - \Exp[X]| > \alpha \Exp[X]]
 \leq 2\exp\!\left(-\frac{\alpha^2 \hat{\rho} }{3 B }\right)
 \leq 2\exp(-(k+1)\ln n) \leq 2/n^{k+1}\,,
\]
where we use that $\hat{\rho} \geq 3(k+1)B\ln n/\alpha^2$. Therefore, with probability at least $1 - 2/n^{k+1}$,
\[
 (1-\alpha) \hat{\rho} / B \leq X \leq (1+\alpha) \hat{\rho} / B
\]
Multiplying everything by $B$ and using that $X  = \rho / B$, we obtain that with probability at least $1 - 2/n^{k+1}$, $(1-\alpha) \hat{\rho} \leq \rho \leq (1+\alpha) \hat{\rho}$, which implies (\ref{eq:rounding}).

We proceed to the case where $\hat{\rho} < 3(k+1)B\ln n/\alpha^2$. Then, assuming that $n$ is large enough that $n^\delta / \ln n > 3(k+1)(q+1)\beta / \alpha^3$, we obtain that $(1-\alpha)\hat{\rho} < (1-\alpha) \alpha n^{q+\delta}$. Therefore, since $\rho \geq 0$, because $\rho_j \geq 0$, for all $j \in N$, the lower bound of (\ref{eq:rounding}) on $\rho$ is trivial.
For the upper bound, we show that with probability at least $1-1/n^{k+1}$, $\rho \leq (1+\alpha) \alpha n^{q+\delta}$. To this end, we consider a sequence $(\rho'_j)_{j \in N}$ so that $\rho_j \leq \rho'_j \leq (q+1)\beta n^q$, for all $j \in N$, and
\[ \hat{\rho}' = \sum_{j \in N} \rho'_{j} y_j = \frac{3(k+1)B\ln n}{\alpha^2} \]
We can obtain such a sequence by increasing an appropriate subset of $\rho_j$ up to $(q+1)\beta n^q$ (if $\sum_{j \in N} \vec{y}$ is not large enough, we may also increase some $y_j$ up to $1$).
For the new sequence, we let $\rho' = \sum_{j \in R} \rho'_{j} z_j$ and observe that $\rho \leq \rho'$, for any instantiation of the randomized rounding (if some $y_j$ are increased, the inequality below follows from a standard coupling argument).
Therefore,
\[ \Prob[\rho > (1+\alpha)\alpha n^{q+\delta}] \leq
   \Prob[\rho' > (1+\alpha)\hat{\rho}']\,,
\]
where we use that $\hat{\rho}' = 3(k+1)B\ln n / \alpha^2$ and that $\alpha n^\delta > 3(k+1)(q+1)\beta \ln n / \alpha^2$, which holds if $n$ is sufficiently large.
By the choice of $\hat{\rho}'$, we can apply the same Chernoff bound as above and obtain that $\Prob[\rho' > (1+\alpha)\hat{\rho}'] \leq 1/n^{k+1}$.
\qed\end{proof}
Lemma~\ref{l:rounding} implies that if the estimations $\rho_j$ are non-negative, the rounded solution $\vec{z}$ is almost feasible for ($d$-IP) with high probability. But, as in Section~\ref{s:pip_sampling}, we need a generalization of Lemma~\ref{l:rounding} that deals with both positive and negative estimations. To this end, we work as in the proof of Lemma~\ref{l:sampling_gen}. Given a sequence of estimations $( \rho_j )_{j \in N}$, with $\rho_j \in [-(q+1)\beta n^q, (q+1)\beta n^q]$, we define $\rho^+_j = \max\{\rho_j, 0\}$ and $\rho^-_j = \min\{ \rho_j, 0\}$, for all $j \in N$. Moreover, we let $\rho^+ = \sum_{j \in N} \rho^+_{j} z_j$, $\hat{\rho}^+ = \sum_{j \in N} \rho^+_{j} y_j$, $\rho^- = \sum_{j \in N} \rho^-_{j} z_j$ and  $\hat{\rho}^- = \sum_{j \in N} \rho^-_{j} y_j$. Applying Lemma~\ref{l:rounding}, once for positive estimations and once for negative estimations (with the absolute values of $\rho_j^-$, $\rho^-$ and $\hat{\rho}^-$, instead), we obtain that with probability at least $1 - 4/n^{k+1}$,
\begin{eqnarray*}
  (1-\alpha)\hat{\rho}^+ - (1-\alpha)\alpha n^{q+\delta} \leq & \rho^+ & \leq
  (1+\alpha)\hat{\rho}^+ + (1+\alpha)\alpha n^{q+\delta} \\
  (1+\alpha)\hat{\rho}^- - (1+\alpha)\alpha n^{q+\delta} \leq & \rho^- & \leq
  (1-\alpha)\hat{\rho}^- + (1-\alpha)\alpha n^{q+\delta}
\end{eqnarray*}
Using that $\rho = \rho^+ + \rho^-$ and that $\hat{\rho} = \hat{\rho}^+ + \hat{\rho}^-$, we obtain the following generalization of Lemma~\ref{l:rounding}.
\begin{lemma}[Rounding Lemma]\label{l:rounding_gen}
Let $\vec{y} \in [0, 1]^n$ be any fractional vector and let $\vec{z} \in \{0, 1\}^n$ be an integral vector obtained from $\vec{y}$ by randomized rounding. Also, let $( \rho_j )_{j \in N}$ be any sequence such that for some integer $q \geq 0$ and some constant $\beta \geq 1$, $|\rho_j| \leq  (q+1)\beta n^q$, for all $j \in N$. For all integers $k \geq 1$ and for all constants $\alpha, \delta > 0$ (and assuming that $n$ is sufficiently large), if $\rho = \sum_{j \in N} \rho_{j} z_j$, $\hat{\rho} = \sum_{j \in N} \rho_{j} y_j$ and $\rb = \sum_{j \in N} |\rho_j|$, with probability at least $1 - 4/n^{k+1}$,
\begin{equation}\label{eq:rounding_gen}
 \hat{\rho} - \alpha\rb - 2\alpha n^{q+\delta} \leq \rho \leq
 \hat{\rho} + \alpha\rb + 2\alpha n^{q+\delta}
\end{equation}
\end{lemma}
For all constants $\e_1, \e_2 > 0$ and all constants $c$, we can use Lemma~\ref{l:rounding_gen} with $\alpha = \max\{\e_1, \e_2/2\}$ and obtain that for all integers $k \geq 1$, with probability at least $1 - 4/n^{k+1}$, the following holds for the binary vector $\vec{z}$ obtained from a fractional vector $\vec{y}$ by randomized rounding.
\begin{equation}\label{eq:pip_rounding2}
 c + \sum_{j \in N} y_j \rho_j -
 \e_1 \overbrace{\sum_{j \in N} |\rho_j|}^{\rb} - \e_2 n^{q+\delta} \leq
 c + \sum_{j \in N} z_j \rho_j
 \leq c + \sum_{j \in N} y_j \rho_j +
 \e_1 \overbrace{\sum_{j \in N} |\rho_j|}^{\rb} + \e_2 n^{q+\delta}
\end{equation}
%
%\begin{equation}\label{eq:pip_rounding2}
% c + \sum_{j \in N} y_j \rho_j -
% \e_1 \rb  - \e_2 n^{q+\delta} \leq
% c + \sum_{j \in N} z_j \rho_j
% \leq c + \sum_{j \in N} y_j \rho_j +
% \e_1 \rb  + \e_2 n^{q+\delta}
%\end{equation}
%
Using (\ref{eq:pip_rounding2}) with $k = 2(d+1)$, the fact that $\vec{y}^\ast$ is a feasible solution to ($d$-LP), and the fact that ($d$-LP) has at most $2n^{d-1}$ constraints, we obtain that $\vec{z}$ is an almost feasible solution to ($d$-IP) with high probability. Namely, with probability at least $1-8/n^{d+4}$, the integral vector $\vec{z}$ obtained from the fractional optimum $\vec{y}^\ast$ by randomized rounding satisfies the following system of inequalities for all levels $\ell \geq 1$ and all tuples $(i_1, \ldots, i_{d-\ell}) \in N^{d-\ell}$ (for each level $\ell \geq 1$, we use $q = \ell - 1$, since $|\rho_{i_1\ldots i_{d-\ell}j}| \leq \ell \beta n^{\ell-1}$ for all $j \in N$).
\begin{equation}\label{eq:pip_deviation}
 c_{i_1\ldots i_{d-\ell}} + \sum_{j \in N} z_j \rho_{i_1\ldots i_{d-\ell}j} \in
 \rho_{i_1\ldots i_{d-\ell}} \pm 2\e_1 \rb_{i_1\ldots i_{d-\ell}}
 \pm  2\e_2 n^{\ell-1+\delta}
\end{equation}
Having established that $\vec{z}$ is an almost feasible solution to ($d$-IP), with high probability, we proceed as in Section~\ref{s:cut_rounding}. By linearity of expectation, $\Exp[ \sum_{j \in N} z_j \rho_j ] = \sum_{j \in V} y^\ast_j \rho_j$. Moreover, the probability that $\vec{z}$ does not satisfy (\ref{eq:pip_deviation}) for some level $\ell \geq 1$ and some tuple $(i_1, \ldots, i_{d-\ell}) \in N^{d-\ell}$ is at most $8/n^{d+4}$ and the objective value of ($d$-IP) is at most $2(d+1)\beta n^d$, because, due to the $\beta$-smoothness property of $p(\vec{x})$, $|p(\vec{x}^\ast)| \leq (d+1)\beta n^d$. Therefore, the expected value of a rounded solution $\vec{z}$ that satisfies the family of inequalities (\ref{eq:pip_deviation}) for all levels and tuples is least $\sum_{j \in V} y^\ast_j \rho_j - 1$ (assuming that $n$ is sufficiently large). Using the method of conditional expectations, as in \cite{Rag88}, we can find in (deterministic) polynomial time an integral solution $\vec{z}$ that satisfies the family of inequalities (\ref{eq:pip_deviation}) for all levels and tuples and has $c + \sum_{j \in V} z_j \rho_j \geq c-1+\sum_{j \in V} y^\ast_j \rho_j$. As in Section~\ref{s:cut_rounding}, we sometimes abuse the notation and refer to such an integral solution $\vec{z}$ (computed deterministically) as the integral solution obtained from $\vec{y}^\ast$ by randomized rounding.

The following lemmas are similar to Lemma~\ref{l:approx} and Lemma~\ref{l:approx_gen}. They use the notion of cumulative absolute value estimations and show that the objective value $p(\vec{z})$ of the rounded solution $\vec{z}$ is close to the optimal value of ($d$-LP).
\begin{lemma}\label{l:approx2}
Let $\vec{y}^\ast$ be an optimal solution of ($d$-LP) and let $\vec{z}$ be the integral solution obtained from $\vec{y}^\ast$ by randomized rounding (and the method of conditional expectations). Then, for any level $\ell \geq 1$ in the decomposition of $p(\vec{x})$ and any tuple $(i_1, \ldots, i_{d-\ell}) \in N^{d-\ell}$,
\begin{equation}\label{eq:approx2}
 p_{i_1\ldots i_{d-\ell}}(\vec{z}) \in
 \rho_{i_1\ldots i_{d-\ell}} \pm 2\e_1 \tb_{i_1\ldots i_{d-\ell}}
 \pm  2\ell \e_2 n^{\ell-1+\delta}
\end{equation}
\end{lemma}
\begin{proof}
The proof is by induction on the degree $\ell$ and similar to the proof of Lemma~\ref{l:approx}. The basis, for $\ell=1$, is trivial, because in the decomposition of $p(\vec{x})$, each $p_{i_1\ldots i_{d}}(\vec{x})$ is a constant $c_{i_1\ldots i_{d}}$\,. Therefore, $\rho_{i_1\ldots i_{d}} = c_{i_1\ldots i_{d}}$ and
\[ p_{i_1\ldots i_{d-1}}(\vec{z}) =
   c + \sum_{j \in N} z_j p_{i_1\ldots i_{d-1}j}(\vec{z})
 = c + \sum_{j \in N} z_j c_{i_1\ldots i_{d-1}j}
 \in \rho_{i_1\ldots i_{d-1}} \pm 2\e_1 \tb_{i_1\ldots i_{d-1}} \pm 2\e_2 n^{\delta}\,,
 \]
where the inclusion follows from the approximate feasibility of $\vec{z}$ for ($d$-LP), as expressed by (\ref{eq:pip_deviation}). We also use that at level $\ell = 1$, $\tb_{i_1\ldots i_{d-1}} = \rb_{i_1\ldots i_{d-1}}$.

We inductively assume that (\ref{eq:approx2}) is true for the values of all degree-$(\ell-1)$ polynomials $p_{i_1\ldots i_{d-\ell}j}$ at $\vec{z}$ and establish the lemma for $p_{i_1\ldots i_{d-\ell}}(\vec{z}) = c_{i_1\ldots i_{d-\ell}} + \sum_{j \in N} z_j p_{i_1\ldots i_{d-\ell}j}(\vec{z})$. We have that:
\begin{align*}
 p_{i_1\ldots i_{d-\ell}}(\vec{z}) & =
 c_{i_1\ldots i_{d-\ell}} + \sum_{j \in N} z_j p_{i_1\ldots i_{d-\ell}j}(\vec{z}) \\
 & \in
 c_{i_1\ldots i_{d-\ell}} + \sum_{j \in N} z_j \left( \rho_{i_1\ldots i_{d-\ell}j}
       \pm 2 \e_1 \tb_{i_1\ldots i_{d-\ell}j}
       \pm 2(\ell-1) \e_2 n^{\ell-2+\delta} \right)\\
 &= \left(c_{i_1\ldots i_{d-\ell}} +
      \sum_{j \in N} z_j \rho_{i_1\ldots i_{d-\ell}j} \right)
      \pm 2 \e_1 \sum_{j \in N} z_j \tb_{i_1\ldots i_{d-\ell}j}
      \pm 2 (\ell-1) \e_2 \sum_{j \in N} z_j n^{\ell-2+\delta} \\
 &\in \left(\rho_{i_1\ldots i_{d-\ell}}
      \pm 2 \e_1 \rb_{i_1\ldots i_{d-\ell}}
      \pm 2 \e_2 n^{\ell-1+\delta}\right)
      \pm 2 \e_1 \sum_{j \in N} \tb_{i_1\ldots i_{d-\ell}j}
      \pm 2 (\ell-1) \e_2 n^{\ell-1+\delta}\\
 &\in \rho_{i_1\ldots i_{d-\ell}} \pm 2 \e_1 \tb_{i_1\ldots i_{d-\ell}}
      \pm 2 \ell \e_2 n^{\ell-1+\delta}
\end{align*}
The first inclusion holds by the induction hypothesis. The second inclusion holds because: (i) $\vec{z}$ is an approximately feasible solution to ($d$-IP) and thus,
$c_{i_1\ldots i_{d-\ell}} + \sum_{j \in N} z_j \rho_{i_1\ldots i_{d-\ell}j}$
satisfies (\ref{eq:pip_deviation}); (ii) $\sum_{j \in N} z_j \tb_{i_1\ldots i_{d-\ell}j} \leq \sum_{j \in N} \tb_{i_1\ldots i_{d-\ell}j}$; and (iii) $\sum_{j \in N} z_j \leq n$. The last inclusion holds because $\tb_{i_1\ldots i_{d-\ell}} = \rb_{i_1\ldots i_{d-\ell}} + \sum_{j \in N} \tb_{i_1\ldots i_{d-\ell}j}$, by the definition of cumulative absolute value estimations.
\qed\end{proof}
\begin{lemma}\label{l:approx2_gen}
Let $\vec{y}^\ast$ be an optimal solution of ($d$-LP) and let $\vec{z}$ be the integral solution obtained from $\vec{y}^\ast$ by randomized rounding (and the method of conditional expectations). Then,
\begin{equation}\label{eq:approx2_gen}
 p(\vec{z}) \in c + \sum_{j \in N} z_j \rho_j
                \pm 2\e_1 \sum_{j \in N} \tb_{j}
                \pm 2 (d-1) \e_2 n^{d-1+\delta}
\end{equation}
\end{lemma}
\begin{proof}
By Lemma~\ref{l:approx2}, for any polynomial $p_j$ appearing in the decomposition of $p(\vec{x})$, we have that $p_j(\vec{z}) \in \rho_j \pm 2 \e_1 \tb_j \pm 2 (d-1) \e_2 n^{d-2+\delta}$. Therefore,
\begin{align*}
 p(\vec{z}) = c + \sum_{j \in N} z_j p_j(\vec{z}) & \in
 c + \sum_{j \in N} z_j \left( \rho_j \pm 2 \e_1 \tb_j
                                \pm 2 (d-1) \e_2 n^{d-2+\delta} \right)\\
 &= c + \sum_{j \in N} z_j \rho_j
      \pm 2 \e_1 \sum_{j \in N} z_j \tb_j
      \pm 2 (d-1) \e_2 \sum_{j \in N} z_j n^{d-2+\delta} \\
 &\in c + \sum_{j \in N} z_j \rho_j
     \pm 2 \e_1 \sum_{j \in N} \tb_j
     \pm 2 (d-1) \e_2 n^{d-1+\delta}
\end{align*}
The second inclusion holds because $z_j \in \{ 0,1\}$ and $\sum_{j \in N} z_j \leq n$.
\qed\end{proof}

\subsection{Cumulative Absolute Value Estimations of $\delta$-Bounded Polynomials}
\label{s:values}

To bound the total error of the algorithm, in Section~\ref{s:pip_together}, we need an upper bound on $\sum_{j \in N} \tb_j$, i.e., on the sum of the cumulative absolute value estimations at the top level of the decomposition of a $\beta$-smooth $\delta$-bounded polynomial $p(\vec{x})$. In this section, we show that $\sum_{j \in N} \tb_j = O(d^2 \beta n^{d-1+\delta})$. This upper bound is an immediate consequence of an upper bound of $O(d\beta n^{d-1+\delta})$ on the sum of the absolute value estimations, for each level $\ell$ of the decomposition of $p(\vec{x})$. 

For simplicity and clarity, we assume, in the statements of the lemmas below and in their proofs, that the hidden constant in the definition of $p(\vec{x})$ as a $\delta$-bounded polynomial is $1$. If this constant is some $\kappa \geq 1$, we should multiply the upper bounds of Lemma~\ref{l:abs_est} and Lemma~\ref{l:cum_est} by $\kappa$. 
\begin{lemma}\label{l:abs_est}
Let $p(\vec{x})$ be an $n$-variate degree-$d$ $\beta$-smooth $\delta$-bounded polynomial. Also let $\rho_{i_1 \ldots i_{d-\ell}}$ and $\rb_{i_1 \ldots i_{d-\ell}}$ be the estimations and absolute value estimations, for all levels $\ell \in \{1, \ldots, d-1\}$ of the decomposition of $p(\vec{x})$ and all tuples $(i_1, \ldots, i_{d-\ell}) \in N^{d-\ell}$, computed by Algorithm~\ref{alg:estimate} and used in ($d$-LP) and ($d$-IP). Then, for each level $\ell \geq 1$, the sum of the absolute value estimations is:
\begin{equation}\label{eq:abs_est}
 \sum_{(i_1, \ldots, i_{d-\ell}) \in N^{d-\ell}} \rb_{i_1\ldots i_{d-\ell}} \leq
 \ell\beta n^{d-1+\delta}
\end{equation}
\end{lemma}
\begin{proof}
The proof is by induction on the level $\ell$ of the decomposition. For the basis, we recall that for $\ell = 1$, level-$1$ absolute value estimations are defined as 
\[ \rb_{i_1\ldots i_{d-1}} = \sum_{j \in N} |\rho_{i_1\ldots i_{d-1} j}|
 = \sum_{j \in N} |c_{i_1\ldots i_{d-1} j}|
\]
This holds because, in Algorithm~\ref{alg:estimate}, each level-$0$ estimation $\rho_{i_1\ldots i_{d-1} i_d}$ is equal to the coefficient $c_{i_1\ldots i_{d-1} i_d}$ of the corresponding degree-$d$ monomial. Hence, if $p(\vec{x})$ is a degree-$d$ $\beta$-smooth $\delta$-bounded polynomial, we have that
\begin{equation}\label{eq:bounded_level1}
 \sum_{(i_1, \ldots, i_{d-1}) \in N^{d-1}} \rb_{i_1\ldots i_{d-1}}
 =  \sum_{(i_1, \ldots, i_{d-1}, j) \in N^{d}} |c_{i_1\ldots i_{d-1} j}|
 \leq \beta n^{d-1+\delta}
\end{equation}
The upper bound holds because by the definition of degree-$d$ $\beta$-smooth $\delta$-bounded polynomials, for each $\ell \in \{ 0, \ldots, d \}$, the sum, over all monomials of degree $d-\ell$, of the absolute values of their coefficients is $O(\beta n^{d-1+\delta})$ (and assuming that the hidden constant is $1$, at most $\beta n^{d-1+\delta}$). In (\ref{eq:bounded_level1}), we use this upper bound for $\ell = 0$ and for the absolute values of the coefficients of all degree-$d$ monomials in the expansion of $p(\vec{x})$.

For the induction step, we consider any level $\ell \geq 2$. We observe that any binary vector $\vec{x}$ satisfies the level-$(\ell-1)$ constraints of ($d$-LP) and ($d$-IP) with certainty, if for each level-$(\ell-1)$ estimation,
\[ 
 \rho_{i_1\ldots i_{d-\ell}j} \leq 
  c_{i_1\ldots i_{d-\ell}j} + \sum_{l \in N} |\rho_{i_1\ldots i_{d-\ell} j l}| =
  c_{i_1\ldots i_{d-\ell}j} + \rb_{i_1\ldots i_{d-\ell} j}
\]
We also note that we can easily enforce such upper bounds on the estimations computed by Algorithm~\ref{alg:estimate}. Since each level-$\ell$ absolute value estimation is defined as $\rb_{i_1\ldots i_{d-\ell}} = \sum_{j \in N} |\rho_{i_1\ldots i_{d-\ell}j}|$, we obtain that for any level $\ell \geq 2$,
\begin{eqnarray*}
 \sum_{(i_1, \ldots, i_{d-\ell}) \in N^{d-\ell}} \rb_{i_1\ldots i_{d-\ell}} & \leq &
 \sum_{(i_1, \ldots, i_{d-\ell}, j) \in N^{d-\ell+1}} \left(|c_{i_1\ldots i_{d-\ell}j}| +
 \rb_{i_1\ldots i_{d-\ell} j} \right)\\
 & \leq & \beta n^{d-1+\delta} + (\ell-1)\beta n^{d-1+\delta}
 = \ell\beta n^{d-1+\delta}
\end{eqnarray*}
For the second inequality, we use the induction hypothesis and that since $p(\vec{x})$ is $\beta$-smooth and $\delta$-bounded, the sum, over all monomials of degree $d-\ell+1$, of the absolute values $|c_{i_1\ldots i_{d-\ell}j}|$  of their coefficients $c_{i_1\ldots i_{d-\ell}j}$ is at most $\beta n^{d-1+\delta}$. We also use the fact that the estimations are computed over the decomposition tree of the polynomial $p(\vec{x})$. Hence, each coefficient $c_{i_1\ldots i_{d-\ell}j}$ is included only once in the sum.
\qed\end{proof}
\begin{lemma}\label{l:cum_est}
Let $p(\vec{x})$ be an $n$-variate degree-$d$ $\beta$-smooth $\delta$-bounded polynomial. Also let $\tb_{i_1 \ldots i_{d-\ell}}$ be the cumulative absolute value estimations, for all levels $\ell \in \{1, \ldots, d-1\}$ of the decomposition of $p(\vec{x})$ and all tuples $(i_1, \ldots, i_{d-\ell}) \in N^{d-\ell}$, corresponding to the estimations $\rho_{i_1 \ldots i_{d-\ell}}$ computed by Algorithm~\ref{alg:estimate} and used in ($d$-LP) and ($d$-IP). Then, 
\begin{equation}\label{eq:cum_est}
 \sum_{j \in N} \tb_{j} \leq d(d-1)\beta n^{d-1+\delta}/2
\end{equation}
\end{lemma}
\begin{proof}
Using induction on the level $\ell$ of the decomposition and Lemma~\ref{l:abs_est}, we show that for each level $\ell \geq 1$, the sum of the cumulative absolute value estimations is:
\begin{equation}\label{eq:cum_est2}
 \sum_{(i_1, \ldots, i_{d-\ell}) \in N^{d-\ell}} \tb_{i_1\ldots i_{d-\ell}} \leq
 (\ell+1)\ell\beta n^{d-1+\delta}/2
\end{equation}
The conclusion of the lemma is obtained by applying (\ref{eq:cum_est2}) for the first level of the decomposition of $p(\vec{x})$, i.e., for $\ell = d-1$.

For the basis, we recall that for $\ell = 1$, level-$1$ cumulative absolute value estimations are defined as
\( \tb_{i_1 \ldots i_{d-1}} = \rb_{i_1 \ldots i_{d-1}} \). Using Lemma~\ref{l:abs_est}, we obtain that:
\[  \sum_{(i_1, \ldots, i_{d-1}) \in N^{d-1}} \tb_{i_1\ldots i_{d-1}} =
 \sum_{(i_1, \ldots, i_{d-1}) \in N^{d-1}} \rb_{i_1\ldots i_{d-1}}
 \leq \beta n^{d-1+\delta}
\]
We recall (see also Section~\ref{s:pip_value}) that for each $\ell \geq 2$, level-$\ell$ cumulative absolute value estimations are defined as
\( \tb_{i_1 \ldots i_{d-\ell}} = \rb_{i_1 \ldots i_{d-\ell}} + \sum_{j \in N} \tb_{i_1 \ldots i_{d-\ell}j} \). 
Summing up over all tuples $(i_1, \ldots, i_{d-\ell}) \in N^{d-\ell}$, we obtain that for any level $\ell \geq 2$,
\begin{eqnarray*}
 \sum_{(i_1, \ldots, i_{d-\ell}) \in N^{d-\ell}} \tb_{i_1\ldots i_{d-\ell}} & = &
 \sum_{(i_1, \ldots, i_{d-\ell}) \in N^{d-\ell}} \left( \rb_{i_1\ldots i_{d-\ell}}
 +
 \sum_{j \in N} \tb_{i_1 \ldots i_{d-\ell}j} \right) \\
 & = & 
 \sum_{(i_1, \ldots, i_{d-\ell}) \in N^{d-\ell}} \rb_{i_1\ldots i_{d-\ell}}
 + \sum_{(i_1, \ldots, i_{d-\ell}, j) \in N^{d-\ell-1}} \tb_{i_1 \ldots i_{d-\ell}j} \\
 & \leq & \ell \beta n^{d-1+\delta} + \ell(\ell-1)\beta n^{d-1+\delta}/2 
 = (\ell+1)\ell\beta n^{d-1+\delta}/2\,,
\end{eqnarray*}
where the inequality follows from Lemma~\ref{l:abs_est} and from the induction hypothesis.
\qed\end{proof}

%\subsection{Concluding the Proof of Theorem~\ref{th:pip_scheme}}\label{s:pip_together}
\subsection{The Final Algorithmic Result}\label{s:pip_together}

We are ready now to conclude this section with the following theorem.
\begin{theorem}\label{th:pip_scheme}
Let $p(\vec{x})$ be an $n$-variate degree-$d$ $\beta$-smooth $\delta$-bounded polynomial. Then, for any $\eps > 0$, we can compute, in time $2^{O(d^7 \beta^3 n^{1-\delta} \ln n/\eps^3)}$ and with probability at least $1-8/n^2$, a binary vector $\vec{z}$ so that $p(\vec{z}) \geq p(\vec{x}^\ast) - \eps n^{d-1+\delta}$, where $\vec{x}^\ast$ is the maximizer of $p(\vec{x})$.
\end{theorem}
\begin{proof}
Based upon the discussion above in this section, for any constant $\eps > 0$, if $p(\vec{x})$ is an $n$-variate degree-$d$ $\beta$-smooth $\delta$-bounded polynomial, the algorithm described in the previous sections computes an integral solution $\vec{z}$ that approximately maximizes $p(\vec{x})$. Specifically, setting $\e_1 = \eps/(4 d(d-1)\beta)$ $\e_2 = \eps/(8(d-1))$, $p(\vec{z})$ satisfies the following with probability at least $1-8/n^2$\,:
\begin{eqnarray*}
 p(\vec{z}) & \geq & \left(c + \sum_{j \in N} y^\ast_j \rho_j\right)
                    - \frac{\eps}{2d(d-1)\beta} \sum_{j \in N} \tb_{j}
                    - \eps n^{d-1+\delta} / 4\\
 & \geq & \left(c + \sum_{j \in N} y^\ast_j \rho_j\right) -
          \eps n^{d-1+\delta} / 2\\
 & \geq & \left(c + \sum_{j \in N} x_j^\ast \rho_j\right) -
          \eps n^{d-1+\delta} / 2\\
 & \geq & \left(p(\vec{x}^\ast) - \frac{\eps}{4d(d-1)\beta} \sum_{j \in N} \tb_{j}
                    - \eps n^{d-1+\delta} / 8\right) -
                    \eps n^{d-1+\delta} / 2\\
 & \geq & p(\vec{x}^\ast) - \eps n^{d-1+\delta}
\end{eqnarray*}
The first inequality follows from Lemma~\ref{l:approx2_gen}. The second inequality follows from the hypothesis that $p(\vec{x})$ is $\beta$-smooth and $\delta$-bounded. Then Lemma~\ref{l:cum_est} implies that $\sum_{j \in N} \tb_{j} \leq \frac{d(d-1)}{2}\beta n^{d-1+\delta}$\,. As in Section~\ref{s:values}, we assume that the constant hidden in the definition of $p(\vec{x})$ as a $\delta$-bounded polynomial is $1$. If this constant is some $\kappa\geq 1$, we should also divide $\e_1$ by $\kappa$. The third inequality holds because $\vec{y}^\ast$ is an optimal solution to ($d$-LP) and $\vec{x}^\ast$ is a feasible solution to ($d$-LP). The fourth inequality follows from Lemma~\ref{l:approx_gen}. For the last inequality, we again use Lemma~\ref{l:cum_est}. This concludes the proof of Theorem~\ref{th:pip_scheme}.\qed\end{proof}

\noindent \textbf{\kCSP}: Using Theorem \ref{th:pip_scheme} it is a
straightforward observation that for any \kCSP\ problem (for constant $k$) we
can obtain an algorithm which, given a \kCSP\ instance with
$\Omega(n^{k-1+\delta})$ constraints for some $\delta>0$, for any $\eps>0$
returns an assignment that satisfies $(1-\eps)\mathrm{OPT}$ constraints in time
$2^{O(n^{1-\delta}\ln n/\eps^3)}$.  This follows from Theorem
\ref{th:pip_scheme} using two observations: first, the standard arithmetization
of \kCSP\ described in Section \ref{s:prelim} produces a degree-$k$
$\beta$-smooth $\delta$-bounded polynomial for $\beta$ depending only on $k$.
Second, the optimal solution of such an instance satisfies at least
$\Omega(n^{k-1+\delta})$ constraints, therefore the additive error given in
Theorem \ref{th:pip_scheme} is $O(\eps \mathrm{OPT})$.  This algorithm for
\kCSP\ contains as special cases algorithm for various standard problems such
as \MC, \MDC\ and \kSAT.

\section{Approximating the \kDense\ in Almost Sparse Graphs}
\label{s:kdense}

In this section, we show how an extension of the approximation algorithms we
have presented can be used to approximate the \kDense\ problem in
$\delta$-almost sparse graphs. Recall that this is a problem also handled in
\cite{AKK99}, but only for the case where $k=\Omega(n)$. The reason that
smaller values of $k$ are not handled by the scheme of \cite{AKK99} for dense
graphs is that when $k=o(n)$ the optimal solution has objective value much
smaller than the additive error of $\eps n^2$ inherent in the scheme.

Here we obtain a sub-exponential time approximation scheme that works on graphs
with $\Omega(n^{1+\delta})$ edges \emph{for all} $k$ by judiciously combining
two approaches: when $k$ is relatively large, we use a sampling approach
similar to \MC; when $k$ is small, we can resort to the na\"ive algorithm that
tries all $n\choose k$ possible solutions. We select (with some foresight) the
threshold between the two algorithms to be $k=\Omega(n^{1-\delta/3})$, so that
in the end we obtain an approximation scheme with running time of
$2^{O(n^{1-\delta/3}\ln n)}$, that is, slightly slower than the approximation
scheme for \MC. It is clear that the brute-force algorithm achieves this
running time for $k=O(n^{1-\delta/3})$, so in the remainder we focus on the
case of large $k$.

The \kDense\ problem in a graph $G(V, E)$ is equivalent to maximizing, over all
binary vectors $\vec{x} \in \{0, 1\}^n$, the $n$-variate degree-$2$ $1$-smooth
polynomial
\( p(\vec{x}) = \sum_{\{i, j\} \in E} x_i x_j \)\,,
under the linear constraint $\sum_{j \in V} x_j = k$. Setting a variable $x_i$ to $1$ indicates that the vertex $i$ is included in the set $C$ that induces a dense subgraph $G[C]$ of $k$ vertices. Next, we assume that $G$ is $\delta$-almost sparse and thus, has $m = \Omega(n^{1+\delta})$ edges. As usual, $\vec{x}$ denotes the optimal solution.

The algorithm follows the same general approach and the same basic steps as the algorithm for \MC\ in Section~\ref{s:maxcut}. In the following, we highlight only the differences. 

\smallskip\noindent{\bf Obtaining Estimations by Exhaustive Sampling.} We first
observe that if $G$ is $\delta$-almost sparse and $k = \Omega(n^{1-\delta/3})$,
then a random subset of $k$ vertices contains $\Omega(n^{1+\delta/3})$ edges in
expectation. Hence, we can assume that the optimal solution induces at least
$\Omega(n^{1+\delta/3})$ edges.

Working as in Section~\ref{s:cut_sampling}, we use exhaustive sampling and
obtain for each vertex $j \in V$, an estimation $\rho_j$ of $j$'s neighbors in
the optimal dense subgraph, i.e., $\rho_j$ is an estimation of $\hat{\rho}_j =
\sum_{i \in N} x_i^\ast$. For the analysis, we apply Lemma~\ref{l:cut_sampling}
with $n^{\delta/3}$, instead of $\Delta$, or in other words, we use a sample of
size $\Theta(n^{1-\delta/3}\ln n)$.  The reason is that we can only tolerate an
additive error of $\eps n^{1+\delta/3}$, by the lower bound on the optimal
solution observed in the previous paragraph.
Then,  the running time due to exhaustive sampling is $ 2^{O(n^{1-\delta/3} \ln
n)}$. 

Thus, by Lemma~\ref{l:cut_sampling} and the discussion following it in
Section~\ref{s:cut_sampling}, we obtain that for all $\e_1, \e_2 > 0$, if we
use a sample of the size $\Theta(n^{1-\delta/3}\ln n /(\e^2_1 \e_2))$, with
probability at least $1 - 2/n^2$, the following holds for all estimations
$\rho_j$ and all vertices $j \in V$:
\begin{equation}\label{eq:dense_sample}
 (1-\e_1)\rho_j - \e_2 n^{\delta/3} \leq \hat{\rho}_j \leq
 (1+\e_1)\rho_j + \e_2 n^{\delta/3}
\end{equation}
%
%As in \MC, we can always assume that $0 \leq \rho_j \leq \deg(j)$, for all $j \in V$.
%
%\smallskip
\noindent{\bf Linearizing the Polynomial.}
Applying Proposition~\ref{pr:decomposition}, we can write the polynomial $p(\vec{x})$ as $p(\vec{x}) = \sum_{j \in V} x_j p_j(\vec{x})$, where $p_j(\vec{x}) = \sum_{i \in N(j)} x_i$ is a degree-$1$ $1$-smooth polynomial that indicates how many neighbors of vertex $j$ are in $C$ in the solution corresponding to $\vec{x}$. Then, using the estimations $\rho_j$ of $\sum_{i \in N(j)} x^\ast_i$\,, obtained by exhaustive sampling, we have that approximate maximization of $p(\vec{x})$ can be reduced to the solution of the following Integer Linear Program:
\begin{alignat*}{3}
& &\max \sum_{j \in V} &y_j \rho_j & & \tag{IP$'$}\\
&\mathrm{s.t.}\quad &
(1-\e_1) \rho_j - \e_2 n^{\delta/3} \leq \sum_{i \in N(j)} &y_i \leq (1+\e_1) \rho_j + \e_2 n^{\delta/3} \quad & \forall &j \in V\\
& & \sum_{i \in N(j)} &y_i = k \\
& & &y_j \in \{0, 1\} &\forall & j \in V
\end{alignat*}
By (\ref{eq:dense_sample}), if the sample size is $|R| = \Theta(n^{1-\delta/3}\ln n/(\e^2_1 \e_2))$, with probability at least $1-2/n^2$, the densest subgraph $\vec{x}^\ast$ is a feasible solution to (IP$'$) with the estimations $\rho_j$ obtained by restricting $\vec{x}^\ast$ to the vertices in $R$. In the following, we let (LP$'$) denote the Linear Programming relaxation of (IP$'$), where each $y_j \in [0, 1]$.

\smallskip\noindent{\bf The Number of Edges in Feasible Solutions.}
We next show that the objective value of any feasible solution $\vec{y}$ to (LP$'$) is close to $p(\vec{y})$. Therefore, assuming that $\vec{x}^\ast$ is feasible, any good approximation to (IP$'$) is a good approximation to the densest subgraph. 
\begin{lemma}\label{l:dense_approx}
Let $\rho_1, \ldots, \rho_n$ be non-negative numbers and $\vec{y}$ be any feasible solution to (LP\,$'$). Then,
\begin{equation}\label{eq:dense_approx}
 p(\vec{y}) \in (1\pm\e_1)\sum_{j \in V} y_j \rho_j \pm \e_2 n^{1+\delta/3}
\end{equation}
\end{lemma}
\begin{proof}
Using the decomposition of $p(\vec{y})$ and the formulation of (LP$'$), we obtain that:
\begin{align*}
 p(\vec{y}) = \sum_{j \in V} y_j \sum_{i \in N(j)} y_i \ & \in
 \sum_{j \in V} y_j \left((1\pm \e_1) \rho_j \pm \e_2 n^{\delta/3}\right) \\
 &= (1\pm \e_1)\sum_{j \in V} y_j \rho_j \pm \e_2 n^{\delta/3} \sum_{j \in V} y_j \\
 &\in (1\pm \e_1)\sum_{j \in V} y_j \rho_j \pm \e_2 n^{1+\delta/3}
\end{align*}
The first inclusion holds because $\vec{y}$ is feasible for (LP$'$) and thus, $\sum_{i \in N(j)} y_i \in (1\pm \e_1)\rho_j \pm \e_2n^{\delta/3}$, for all $j$. The second inclusion holds because $\sum_{j \in V} y_j \leq n$.
\qed\end{proof}
\noindent{\bf Randomized Rounding of the Fractional Optimum.}
As a last step, we show how to round the fractional optimum $\vec{y}^\ast =
(y^\ast_1, \ldots, y^\ast_n)$ of (LP$'$) to an integral solution $\vec{z} =
(z_1, \ldots, z_n)$ that almost satisfies the constraints of (IP$'$).  To this
end, we use randomized rounding, as for \MC. We obtain that with probability at
least $1 - 2/n^{8}$,
\begin{equation}\label{eq:k_deviation}
 k - 2\sqrt{n\ln(n)} \leq
 \sum_{j \in V} z_i  \leq
 k + 2\sqrt{n\ln(n)}
\end{equation}
Specifically, the inequality above follows from the Chernoff bound in footnote~\ref{foot:chernoff}, with $t = 2\sqrt{n \ln(n)}$, since $\Exp[\sum_{i \in N(j)} z_j] = k$. 
Moreover, applying Lemma~\ref{l:rounding} with $q = 0$, $\beta = 1$, $k = 7$, $\delta/3$ (instead of $\delta$) and $\alpha = \max\{ \e_1, \e_2/2\}$, and using that $\vec{y}^\ast$ is a feasible solution to (LP$'$) and that $\e_1 \in (0, 1)$, we obtain that with probability at least $1 - 2/n^{8}$, for each vertex $j$,
\begin{equation}\label{eq:z_deviation}
 (1-\e_1)^2\rho_j - 2\e_2 n^{\delta/3} \leq
 \sum_{i \in N(j)} z_i  \leq
 (1+\e_1)^2\rho_j + 2\e_2 n^{\delta/3} 
\end{equation}
By the union bound, the integral solution $\vec{z}$ obtained from $\vec{y}^\ast$ by randomized rounding satisfies (\ref{eq:k_deviation}) and (\ref{eq:z_deviation}), for all vertices $j$, with probability at least $1 - 3/n^7$.

By linearity of expectation, $\Exp[ \sum_{j \in V} z_j \rho_j ] = \sum_{j \in V} y^\ast_j \rho_j$. Moreover, since the probability that $\vec{z}$ does not satisfy
either (\ref{eq:k_deviation}) or (\ref{eq:z_deviation}), for some vertex $j$, is at most $3/n^7$, and since the objective value of (IP$'$) is at most $n^2$, the expected value of a rounded solution $\vec{z}$ that (\ref{eq:k_deviation}) and (\ref{eq:z_deviation}), for all vertices $j$, is least $\sum_{j \in V} y^\ast_j \rho_j - 1$ (assuming that $n \geq 2$). As in \MC, such an integral solution $\vec{z}$ can be found in (deterministic) polynomial time using the method of conditional expectations (see \cite{Rag88}). 

The following is similar to Lemma~\ref{l:dense_approx} and shows that the objective value $p(\vec{z})$ of the rounded solution $\vec{z}$ is close to the optimal value of (LP$'$). 
\begin{lemma}\label{l:dense_approx2}
Let $\vec{y}^\ast$ be the optimal solution of (LP$'$) and let $\vec{z}$ be the integral solution obtained from $\vec{y}^\ast$ by randomized rounding (and the method of conditional expectations). Then,
\begin{equation}\label{eq:dense_approx2}
 p(\vec{z}) \in (1 \pm \e_1)^2 \sum_{j \in V} y^\ast_j \rho_j \pm 3\e_2 n^{1+\delta/3}
\end{equation}
\end{lemma}
\begin{proof}
Using the decomposition of $p(\vec{y})$ and an argument similar to that in the proof of Lemma~\ref{l:dense_approx}, we obtain that:
\begin{align*}
 p(\vec{z}) = \sum_{j \in V} z_j \sum_{i \in N(j)} z_i \ \ & \in 
 \sum_{j \in V} z_j \left((1\pm \e_1)^2 \rho_j \pm 2\e_2 n^{\delta/3} \right)  \\
 &= (1\pm \e_1)^2 \sum_{j \in V} z_j \rho_j 
    \pm 2 \e_2 n^{\delta/3} \sum_{j \in V} z_j\\
 &\in (1\pm \e_1)^2 \sum_{j \in V} z_j \rho_j \pm 2\e_2 n^{1+\delta/3} \\
 &\in (1\pm \e_1)^2 \sum_{j \in V} y^\ast_j \rho_j \pm 3\e_2 n^{1+\delta/3}
\end{align*}
The first inclusion holds because $\vec{z}$ satisfies (\ref{eq:z_deviation}) for all $j \in V$. For the second inclusion, we use that $\sum_{j \in V} z_j \leq n$. For the last inclusion, we recall that $\sum_{j \in V} z_j \rho_j \geq \sum_{j \in V} y^\ast_j \rho_j - 1$ and assume that $n$ is sufficiently large.
\qed \end{proof}
\noindent{\bf Putting Everything Together.}
Therefore, for $\eps > 0$, if $G$ is $\delta$-almost sparse and $k =
\Omega(n^{1-\delta/3})$, the algorithm described computes estimations $\rho_j$
such that the densest subgraph $\vec{x}^\ast$ is a feasible solution to (IP$'$)
whp.  Hence, by the analysis above, the algorithm computes a slightly
infeasible solution approximating the number of edges in the densest subgraph
with $k$ vertices within a multiplicative factor of $(1-\e_1)^2$ and an
additive error of $\e_2 n^{1+\delta/3}$. Setting $\e_1 = \e_2 = \eps/8$, the
number of edges in the subgraph induced by $\vec{z}$ satisfies the following
with probability at least $1-2/n^2$\,:
\[ p(\vec{z}) \geq (1-\e_1)^2 \sum_{j \in V} y_j^\ast \rho_j - 3 \e_2 n^{1+\delta/3}               \geq 
	(1-\e_1)^2 \sum_{j \in V} x_j^\ast \rho_j - 3 \e_2 n^{1+\delta/3} \geq
	p(\vec{x}^\ast) - \eps n^{1+\delta/3} \geq
	(1-\eps) p(\vec{x}^\ast)
%
%\sum_{j \in V} x_j^\ast (\deg(j) - \rho_j) - 3 \eps m/8
%              \geq p(\vec{x}^\ast) - \eps m / 2 \geq (1-\eps) p(\vec{x}^\ast)
\]
The first inequality follows from Lemma~\ref{l:dense_approx2}, the second
inequality holds because $\vec{y}^\ast$ is the optimal solution to (LP) and
$\vec{x}^\ast$ is feasible for (LP), the third inequality follows from
Lemma~\ref{l:dense_approx} and the fourth inequality holds because the optimal
cut has at least $\Omega(n^{1+\delta/3})$ edges.

This solution is infeasible by at most $2\sqrt{n \ln n}=o(k)$ vertices and can
become feasible by adding or removing at most so many vertices and
$O(n^{1/2+\delta})$ edges. 
\begin{theorem}\label{th:densest}
Let $G(V, E)$ be a $\delta$-almost sparse graph with $n$ vertices. Then, for any integer $k \geq 1$ and for any $\eps > 0$, we can compute, in time $2^{O(n^{1-\delta/3} \ln n/\eps^3)}$ and with probability at least $1-2/n^2$, an induced subgraph $\vec{z}$ of $G$ with $k$ vertices whose number of edges satisfies $p(\vec{z}) \geq (1-\eps)p(\vec{x}^\ast)$, where $\vec{x}^\ast$ is the number of edges in the \kDense\ of $G$. 
\end{theorem}

\section{Lower Bounds} \label{s:lower}

%Two items:
% 1. Approach does not work for Max-k-SAT, k>2
% 2. Approach is tight for MAX-CUT

In this section we give some lower bound arguments which show that the
algorithmic schemes we have presented are, in some senses, likely to be almost
optimal.  Our working complexity assumption will be the Exponential Time
Hypothesis (ETH), which states that there is no algorithm that can solve an
instance of 3-SAT of size $n$ in time $2^{o(n)}$.

Our starting point is the following inapproximability result,
which can be obtained using known PCP constructions and standard reductions.
\begin{theorem} \label{thm:start}
There exist constants $c,s\in[0,1]$ with $c>s$ such that for all $\epsilon>0$
we have the following: if there exists an algorithm which, given an $n$-vertex
$5$-regular instance of \MC, can distinguish between the case where a solution
cuts at least a $c$ fraction of the edges and the case where all solutions cut
at most an $s$ fraction of the edges in time $2^{n^{1-\epsilon}}$ then the ETH
fails.
\end{theorem}
\begin{proof}
This inapproximability result follows from the construction of quasi-linear
size PCPs given, for example, in \cite{Dinur05}. In particular, we use as
starting point a result explicitly formulated in \cite{MR10} as follows:
``Solving 3-\textsc{SAT} on inputs of size $N$ can be reduced to distinguishing
between the case that a 3CNF formula of size $N^{1+o(1)}$ is satisfiable and
the case that only $\frac{7}{8} + o(1)$ fraction of its clauses are
satisfiable''.

Take an arbitrary 3-\textsc{SAT} instance of size $N$, which according to the
ETH cannot be solved in time $2^{o(N)}$. By applying the aforementioned PCP
construction we obtain a 3CNF formula of size $N^{1+o(1)}$ which is either
satisfiable or far from satisfiable. Using standard constructions
(\cite{PY91,BK99}) we can reduce this formula to a $5$-regular graph $G(V,E)$
which will be a \MC\ instance (we use degree $5$ here for concreteness, any
reasonable constant would do).  We have that $|V|$ is only a constant factor
apart from the size of the 3CNF formula.  At the same time, there exist
constants $c,s$ such that, if the formula was satisfiable $G$ has a cut of
$c|E|$ edges, while if the formula was far from satisfiable $G$ has no cut with
more than $s|E|$ edges. If there exists an algorithm that can distinguish
between these two cases in time $2^{|V|^{1-\epsilon}}$ the whole procedure
would run in $2^{N^{1-\epsilon+o(1)}}$ and would allow us to decide if the
original formula was satisfiable.
\qed\end{proof}
There are two natural ways in which one may hope to improve or extend the
algorithms we have presented so far: relaxing the density requirement or
decreasing the running time. We prove in what follows that none of them can improve the results presented so far.

\subsection{Arity Higher Than Two}

First, recall that the algorithm we have given for
\kCSP\ works in the density range between $n^k$ and $n^{k-1}$.  Here, we give a
reduction establishing that it's unlikely that this can be improved.
\begin{theorem} \label{thm:hard1}
There exists $r>1$ such that for all $\epsilon>0$ and all (fixed) integers
$k\ge 3$ we have the following: if there exists an algorithm which approximates
\textsc{Max-$k$-SAT} on instances with $\Omega(n^{k-1})$ clauses in time
$2^{n^{1-\epsilon}}$ then the ETH fails.
\end{theorem}
\begin{proof} %[Theorem \ref{thm:hard1}]
Consider the \MC\ instance of Theorem \ref{thm:start}, and transform it into a
2-\textsc{SAT} instance in the standard way: the set of variables is the set of
vertices of the graph and for each edge $(u,v)$ we include the two clauses
$(\neg u \lor v)$ and $(u\lor \neg v)$. This is an instance of 2-\textsc{SAT}
with $n$ variables and $5n$ clauses and there exist constants $c,s$ such that
either there exists an assignment satisfying a $c$ fraction of the clauses or
all assignments satisfy at most an $s$ fraction of the clauses.

Fix a constant $k$ and introduce to the instance $(k-2)n$ new variables
$x_{(i,j)}$, $i\in\{1,\ldots,k-2\}$, $j\in\{1,\ldots,n\}$. We perform the
following transformation to the 2-\textsc{SAT} instance: for each clause
$(l_1\lor l_2)$ and for each tuple
$(i_1,i_2,\ldots,i_{k-2})\in\{1,\ldots,n\}^{k-2}$ we construct $2^{k-2}$ new
clauses of size $k$. The first two literals of these clauses are always $l_1,
l_2$. The remaining $k-2$ literals consist of the variables
$x_{(1,i_1)},x_{(2,i_2)},\ldots,x_{(k,i_{k-2})}$, where in each clause we pick
a different set of variables to be negated. In other words, to construct a
clause of the new instance we select a clause of the original instance, one
variable from each of the $(k-2)$ groups of $n$ new variables, and a subset of
these variables that will be negated. The new instance consists of all the size
$k$ clauses constructed in this way, for all possible choices.

First, observe that the new instance has $5n^{k-1}2^k$ clauses and $(k-1)n$
variables, therefore, for each fixed $k$ it satisfies the density conditions of
the theorem. Furthermore, consider any assignment of the original formula. Any
satisfied clause has now been replaced by $2^k$ satisfied clauses, while for an
unsatisfied clause any assignment to the new variables satisfies exactly
$2^k-1$ clauses. Thus, for fixed $k$, there exist constants $s',c'$ such that
either a $c'$ fraction of the clauses of the new instance is satisfiable or at
most a $s'$ fraction is. If there exists an approximation algorithm with ratio
better than $c'/s'$ running in time $2^{N^{1-\epsilon}}$, where $N$ is the
number of variables of the new instance, we could use it to decide the original
instance in a time bound that would disprove the ETH.
\qed\end{proof}

\subsection{Almost Tight Time Bounds}

A second possible avenue for improvement may be to consider potential speedups
of our algorithms. Concretely, one may ask whether the (roughly)
$2^{\sqrt{n}\ln n}$ running time guaranteed by our scheme for \MC\ on graphs
with average degree $\sqrt{n}$ is best possible. We give an almost tight answer
to such questions via the following theorem.
\begin{theorem} \label{thm:hard2}
There exists $r>1$ such that for all $\epsilon>0$ we have the following: if
there exists an algorithm which, for some $\Delta=o(n)$, approximates \MC\ on
$n$-vertex $\Delta$-regular graphs in time $2^{(n/\Delta)^{1-\epsilon}}$ then
the ETH fails.
\end{theorem}
\begin{proof}[Theorem \ref{thm:hard2}]
Without loss of generality we prove the theorem for the case when the degree is
a multiple of 10.

Consider an instance $G(V,E)$ of \MC\ as given by Theorem \ref{thm:start}. Let
$n=|V|$ and suppose that the desired degree is $d=10\Delta$, where $\Delta$ is
a function of $n$.  We construct a graph $G'$ as follows: for each vertex $u\in
V$ we introduce $\Delta$ new vertices $u_1,\ldots,u_\Delta$ as well as
$5\Delta$ ``consistency'' vertices $c^u_1,\ldots,c^u_{5\Delta}$. For every edge
$(u,v)\in E$ we add all edges $(u_i,v_j)$ for $i,j\in\{1,\ldots,\Delta\}$.
Also, for every $u\in V$ we add all edges $(u_i,c^u_j)$, for
$i\in\{1,\ldots,\Delta\}$ and $j\in\{1,\ldots,5\Delta\}$. This completes the
construction.

The graph we have constructed is $10\Delta$-regular and is made up of $6\Delta
n$ vertices. Let us examine the size of its optimal cut. Consider an optimal
solution and observe that, for a given $u\in V$ all the vertices $c^u_i$ can be
assumed to be on the same side of the cut, since they all have the same
neighbors. Furthermore, for a given $u\in V$, all vertices $u_i$ can be assumed
to be on the same side of the cut, namely on the side opposite that of $c^u_i$,
since the vertices $c^u_i$ are a majority of the neighborhood of each $u_i$.
With this observation it is easy to construct a one-to-one correspondence
between cuts in $G$ and locally optimal cuts in $G'$.

Consider now a cut that cuts $c|E|$ edges of $G$. If we set all $u_i$ of $G'$
on the same side as $u$ is placed in $G$ we cut $c|E|\Delta^2$ edges of the
form $(u_i,v_j)$. Furthermore, by placing the $c^u_i$ on the opposite side of
$u_i$ we cut $5\Delta^2 |V|$ edges. Thus the max cut of $G'$ is at least
$c|E|\Delta^2 + 5\Delta^2 |V|$. Using the previous observations on locally
optimal cuts of $G'$ we can conclude that if $G'$ has a cut with $s|E|\Delta^2
+ 5\Delta^2|V|$ edges, then $G$ has a cut with $s|E|$ edges. Using the fact
that $2|E|=5|V|$ (since $G$ is 5-regular) we get a constant ratio between the
size of the cut of $G'$ in the two cases. Call that ratio $r$.

Suppose now that we have an approximation algorithm with ratio better than $r$
which, given an $N$-vertex $d$-regular graph runs in time
$2^{(N/d)^{1-\epsilon}}$. Giving our constructed instance as input to this
algorithm would allow to decide the original instance in time
$2^{n^{1-\epsilon}}$.
\qed\end{proof}
Theorem \ref{thm:hard2} establishes that our approach is essentially
optimal, not just for average degree $\sqrt{n}$, but for any other intermediate
density.

%\section{Conclusions and Directions for Further Research}
%
%Our basic result in this paper is that \kCSP\ instances on~$n$ variables with $\Omega(n^{k-1+\delta})$ constraints can be solved within approximation ratio
%$(1-\eps)$  in time $2^{O(n^{1-\delta}\ln n
%/\eps^3)}$, for any fixed $\delta\in (0,1]$, $\eps>0$ and integer $k\ge 2$. As a consequence, \MC{} admits an approximation scheme running in
%time $2^{O(\frac{n}{\Delta}\ln n/\eps^3)}$ for graphs with average degree
%$\Delta$. 
%%In other words, this is an approximation scheme that runs in time \emph{sub-exponential in $n$} even for almost sparse instances where the average degree is $\Delta = n^\delta$ for some small $\delta>0$. 
%More generally, for graph problems the trade-off provided by our basic result covers the whole spectrum from dense to almost sparse instances,
%while for general \kCSP, it covers instances where the number of constraints ranges from $\Theta(n^{k})$ to $\Theta(n^{k-1})$. Furthermore, we have shown that our results are essentially best possible, assuming the ETH.
%
%There is a number of interesting open questions motivated by our results. There exist numerous problems, the best polynomial approximation ratios of which come from LP-based techniques (associated with several smart ideas of randomized rounding, the \textsc{Steiner Tree} , the \textsc{Multiway Cut}, several versions of location problem, etc.) At what extent, our techniques, or the development of new ad-hoc methods are able to derive tight subexponential approximation results for some of them?
%

\bibliographystyle{plain}
\bibliography{subexponential}

%\appendix
%\input{appendix}
%\input{general}
%\input{rounding}
%\input{dense2}
%\input{lower2}

\end{document}